\documentclass[opre]{informs4}

\IfFileExists{tgtermes.sty}{\RequirePackage{tgtermes}}{}
\IfFileExists{newtxtext.sty}{\RequirePackage{newtxtext}}{}
\IfFileExists{newtxmath.sty}{\RequirePackage{newtxmath}}{}
\RequirePackage{bm}
\RequirePackage{endnotes}

\OneAndAHalfSpacedXI

\newcommand{\CompactDisplaySpacing}{%
  \setlength{\abovedisplayskip}{5pt}%
  \setlength{\belowdisplayskip}{5pt}%
  \setlength{\abovedisplayshortskip}{0pt}%
  \setlength{\belowdisplayshortskip}{4pt}%
}

\usepackage{booktabs}
\usepackage{comment}
\usepackage{setspace}

\providecommand{\E}{\mathbb E}
\providecommand{\Op}{O_p}
\providecommand{\Hcal}{\mathcal H}
\providecommand{\norm}[1]{\left\lVert #1\right\rVert}

\usepackage{natbib}
\bibpunct[, ]{(}{)}{,}{a}{}{,}%

\def\bibfont{\SingleSpacedXI\small}
\usepackage[hidelinks]{hyperref}

\makeatletter
\def\theARTICLEABSTRACT{%
  \HOOKb
  \vspace*{18pt}%
  \noindent\begin{minipage}[t]{\textwidth}%
    \parindent1em\ABSfont
    \noindent\theABSTRACT\endgraf
    \vskip5pt
    \theKEYWORDS
    \noindent\hrulefill
  \end{minipage}%
  \vspace*{0pt}%
}
\makeatother

\EquationsNumberedThrough
\TheoremsNumberedThrough
\ECRepeatTheorems

\MANUSCRIPTNO{}

\def\theARTICLETOP{}
\def\theARTICLETOPLEFT{}
\def\theARTICLETOPRIGHT{}
\JOURNAL{}
\JOURNALshort{}
\RRHFirstLine{}
\LRHFirstLine{}
\RRHSecondLine{\bf\theRUNAUTHOR:\enskip {\it\theRUNTITLE}}
\LRHSecondLine{\bf\theRUNAUTHOR:\enskip {\it\theRUNTITLE}}
\ECRRHFirstLine{}
\ECLRHFirstLine{}
\ECRRHSecondLine{\bf\theRUNAUTHOR:\enskip {\it\theRUNTITLE}}
\ECLRHSecondLine{\bf\theRUNAUTHOR:\enskip {\it\theRUNTITLE}}
\RUNAUTHOR{Lin, Zhang and Hong}

\begin{document}

\RUNTITLE{Learn-Then-Differentiate Gradient Estimation}
\TITLE{Learn-Then-Differentiate Gradient Estimation}

\ARTICLEAUTHORS{%
\AUTHOR{Nifei Lin}
\AFF{Research Institute for Interdisciplinary Sciences, School of Information Management and Engineering, Shanghai University of Finance and Economics, Shanghai 200433, China}
\AUTHOR{Qingkai Zhang}
\AFF{Department of Decision Analytics and Operations, City University of Hong Kong, Hong Kong, China}
\AUTHOR{L. Jeff Hong}
\AFF{Department of Industrial and Systems Engineering, University of Minnesota, Minneapolis, Minnesota 55455}
}%

\ABSTRACT{%
Learn-then-differentiate (LTD) estimates gradients by fitting a model to simulation outputs and differentiating it. We develop a unified framework explaining what LTD differentiates and how accurately it estimates gradients. For models with a weighted representation, LTD differentiates a learned representation of the underlying probability measure. We then show how accuracy guarantees for fitted models translate into guarantees for gradients and higher-order derivatives, with rates approaching the standard Monte Carlo rate under suitable smoothness conditions. The framework recovers established results for kernel regression, local polynomial regression, and kernel ridge regression, and yields further guarantees for multiple kernel learning and smooth neural networks. These results provide a common foundation for understanding and analyzing LTD across learning methods.
}%


\maketitle

\newtoks\MainTextDisplayHook
\MainTextDisplayHook=\expandafter{\the\everydisplay}
\everydisplay=\expandafter{\the\everydisplay
  \CompactDisplaySpacing
  \setlength{\jot}{5pt}%
}

\section{Introduction}
\label{sec:introduction}

Gradient information plays a central role in simulation optimization, sensitivity analysis, uncertainty quantification, and financial risk management. In queueing and service systems, gradients quantify the sensitivity of waiting times, throughput, congestion, and service levels to arrival rates, service rates, staffing levels, and routing parameters. In inventory and supply-chain systems, derivatives with respect to order quantities, reorder levels, lead times, demand parameters, and pricing decisions support the optimization of operational policies. In financial engineering, sensitivities of option prices and risk measures to market parameters, commonly known as the Greeks, are indispensable for hedging and risk management. Similar applications arise in reliability, manufacturing, communication networks, healthcare operations, and other simulation-based decision problems \citep{hocao1991,glasserman2004,fu2006}.

Motivated by these applications, estimating gradients of expected system performance has been a fundamental problem in stochastic simulation for nearly five decades \citep{ho1979,hocao1983,reiman1989,glasserman1991,fuhu1997,liu2011,peng2018}. \citet{glasserman2004} provides a unifying perspective by viewing expected performance as the integral of a performance function with respect to a parameterized probability measure. The parameter may affect the performance function, the probability measure, or both. This representation allows simulation gradient estimators to be classified according to the mathematical object on which differentiation acts.

The first class differentiates the sample performance while keeping the probability measure fixed. Representative examples include finite-difference methods, infinitesimal perturbation analysis, pathwise differentiation, conditional Monte Carlo, and kernel estimators \citep{hocao1991,glasserman1991,fuhu1997,liu2011}. These methods use different analytical or statistical techniques to differentiate, or approximate the differentiation of, the sample performance under a fixed probability measure.

The second class differentiates the probability measure while leaving the performance function unchanged. Likelihood-ratio or score-function methods do so by differentiating the probability density or likelihood ratio \citep{reiman1989,glynn1990,rubinstein1993}. Weak-derivative and measure-valued differentiation methods instead represent the derivative of the probability measure as a difference of two measures \citep{pflug1996,heidergott2006}. Despite their different formulations, these methods share the principle of differentiating the probability measure.

The third class combines differentiation of the sample performance and the probability measure. Examples include generalized likelihood-ratio methods for discontinuous performance measures \citep{peng2018,peng2020} and Malliavin-calculus-based estimators developed for financial engineering \citep{fournie1999,chen2007}. Although these estimators have substantially different derivations, Glasserman's framework identifies their common use of both sources of parameter dependence. Together, the three classes organize a diverse literature around the mathematical object being differentiated.

Advances in machine learning have brought renewed attention to an alternative route to gradient estimation: learning an approximation of the expected-performance function from simulation data and then differentiating the learned surrogate. We refer to this strategy as \emph{learn-then-differentiate} (LTD). The surrogate may be constructed using higher-order kernel regression \citep{hansen2008}, local polynomial regression \citep{neumeyer2010,racine2016}, kernel ridge regression \citep{fischer2020,liu2023derivatives}, and other learning algorithms. Convergence rates for LTD estimators have also been studied, with existing analyses developed separately for specific learning methods or surrogate models \citep{hansen2008,neumeyer2010,fischer2020,liu2023derivatives}.

LTD offers several features that are attractive for simulation applications. First, it requires little structural information about the simulator: it does not require differentiable sample paths, as pathwise methods do, or analytical knowledge of the underlying probability model, as likelihood-ratio methods do. It therefore applies naturally to black-box simulators whose internal mechanisms are inaccessible or difficult to analyze. Second, LTD constructs a differentiable approximation over a region of the parameter space, allowing the same surrogate to provide gradient estimates at multiple parameter values. This feature is useful when evaluation points are not known in advance, as in simulation optimization. Third, LTD separates offline simulation and surrogate construction from online gradient evaluation. Once learned, the surrogate can be differentiated rapidly, supporting real-time decision making.

This separation connects LTD to the growing literature on simulation learning, which integrates statistical learning with stochastic simulation to support offline simulation and online decision making. \citet{hongjiang2019} describe this philosophy as the \emph{offline-simulation-online-application} (OSOA) paradigm. Representative applications include simulation metamodeling \citep{ankenman2010,barton2015}, online risk monitoring \citep{jiang2020}, contextual ranking and selection \citep{shen2021,du2024,keslin2025}, and contextual simulation optimization \citep{jin2026,lin2025}. LTD fits this paradigm by learning the expected-performance function offline and evaluating its derivatives online. Moreover, whereas statistical learning from observed data typically takes the design points as given, simulation learning allows the experimenter to choose design points and allocate simulation effort.

Despite these appealing features and the available results for individual methods, a unified understanding of LTD requires addressing two central questions. The first is \emph{What does LTD differentiate?} At the computational level, the answer is the learned surrogate. Within Glasserman's framework, however, its interpretation is less immediate: LTD directly differentiates neither the sample performance nor an analytically specified probability measure. The second is \emph{How statistically efficient is LTD?} In particular, how does the approximation accuracy of the learned surrogate translate into a convergence rate for derivative estimation? Answering these questions across learning algorithms requires a common mathematical structure connecting surrogate construction, differentiation, and statistical analysis.

This paper develops such a structure through a unified weighted framework for LTD gradient estimation. Our starting point is that a broad class of LTD estimators admits a common weighted representation. Within this class, learning algorithms differ in how they construct the weights while sharing the same mathematical structure for derivative estimation. Rather than developing a separate estimator and analysis for each surrogate model, the framework allows the corresponding LTD estimator and its statistical properties to be obtained systematically from the learning method and its approximation guarantees. Its scope extends beyond the representative methods studied in this paper to other learning algorithms and surrogate models satisfying the framework's conditions.

The main contributions are as follows.
\begin{enumerate}
\item \textit{A unified weighted representation.}
We establish a common weighted representation for a broad class of LTD estimators. This representation provides an algorithm-independent foundation for interpreting and analyzing surrogate-based gradient estimation.

\item \textit{An interpretation through measure differentiation.}
We answer ``{What does LTD differentiate?}" by showing that every weighted LTD estimator induces a learned measure representation constructed from simulation observations. Differentiating the learned surrogate is mathematically equivalent to differentiating this learned measure representation. LTD therefore belongs to the measure-differentiation class in Glasserman's framework, with the measure represented nonparametrically through statistical learning rather than analytically through an explicit probabilistic model.

\item \textit{A distinction from likelihood-ratio estimation.}
We show that LTD constructs the learned measure representation directly on the random objects of interest, avoiding the state-space enlargement that likelihood-ratio formulations may require. Our numerical experiments demonstrate that avoiding this enlargement can yield substantially lower variance than likelihood-ratio estimators.

\item \textit{A general convergence theory.}
We answer ``{How statistically efficient is LTD?}" by developing an interpolation principle that transfers convergence rates for surrogate learning to convergence rates for derivative estimation. The theory applies broadly to differentiable surrogate models satisfying the required conditions and shows how the smoothness of the expected-performance function and the approximation capability of the learning algorithm jointly determine statistical efficiency. In particular, by relating the smoothness of the expected-performance function to that of the underlying probability measure, we show that LTD estimators can achieve convergence rates arbitrarily close to the canonical $n^{-1/2}$ rate under sufficiently smooth probabilistic models.

\item \textit{Applications across learning methods.}
We illustrate the general convergence theory using five representative learning methods. For kernel regression, local polynomial regression, and kernel ridge regression, the framework recovers established derivative-rate results in the corresponding settings through a common argument. For multiple kernel learning and a particular type of smooth neural network method, the framework yields new convergence guarantees for gradients and higher-order derivatives, under the stated regularity conditions. Together, these applications demonstrate that the framework both unifies existing results and provides a systematic route to new derivative-estimation guarantees. 
\end{enumerate}

The remainder of the paper is organized as follows. Section~\ref{sec:foundations} reviews Glasserman's conceptual framework for simulation gradient estimation. Section~\ref{sec:ltd} formulates the general LTD paradigm, establishes the weighted representation for LTD estimators and introduces representative surrogate-learning methods. Section~\ref{sec:ltd-interpretation} proves the representation theorem, interprets LTD as differentiation of an induced measure, and examines its relationship to classical likelihood-ratio methods. Section~\ref{sec:ltd-rate} develops the general convergence theory and applies it to the representative learning methods. Section~\ref{sec:numerical} presents numerical experiments, and Section~\ref{sec:conclusion} concludes.

\section{Foundations}
\label{sec:foundations}

Consider a stochastic simulation model with parameter $\theta\in\mathbb{R}^d$. Let $\xi$ denote the random vector representing all sources of stochasticity in the simulation and let $F(\xi)$ denote the resulting simulation output, where the dependence on $\theta$ is suppressed for notational simplicity. The expected performance is
\[
J(\theta)=\mathbb{E}_P[F(\xi)],
\]
where the parameter $\theta$ may enter the performance function $F$, the probability measure $P$ of $\xi$, or both. Our objective is to estimate the gradient $\nabla J(\theta)$, or more generally the $r$th-order derivative $\nabla^{(r)}J(\theta)$ for a positive integer $r$, using information obtained from stochastic simulation experiments. Although the classical simulation gradient estimation literature typically considers a scalar parameter ($d=1$), since multivariate gradients can be obtained componentwise by holding the remaining coordinates fixed, we explicitly allow $\theta$ to be $d$-dimensional throughout the paper.

\subsection{Glasserman's Framework}
\label{subsec:glasserman}

Viewing the expectation as an integral provides a unified perspective on simulation gradient estimation. Since
\[
J(\theta)=\int F(\theta,\xi)\,dP_\theta(\xi),
\]
the parameter $\theta$ enters the expectation through two mathematical objects: the performance function $F$ and the probability measure $P_\theta$. Consequently, estimating derivatives of $J(\theta)$ amounts to differentiating one or both of these objects. \citet{glasserman2004} organized simulation gradient estimation methods according to this principle, leading to two broad classes of approaches: differentiating the performance function and differentiating the probability measure.

The first class differentiates the performance function while keeping the probability measure fixed. Assuming that the dependence on $\theta$ is entirely through the performance function and that differentiation and integration may be interchanged,
\[
\nabla J(\theta)
=
\nabla\int F(\theta,\xi)\,dP(\xi)
=
\int \nabla_\theta F(\theta,\xi)\,dP(\xi)
=
\mathbb{E}[\nabla_\theta F(\theta,\xi)].
\]
This immediately yields the familiar pathwise (or infinitesimal perturbation analysis) estimator
\[
\widehat{\nabla J}_{\rm PW}
=
\frac{1}{n}\sum_{i=1}^n\nabla_\theta F(\theta,\xi_i).
\]

The second class differentiates the probability measure while keeping the performance function fixed. Assuming that the dependence on $\theta$ is entirely through the probability measure,
\[
\nabla J(\theta)
=
\nabla\int F(\xi)\,dP_\theta(\xi)
=
\int F(\xi)\,d(\nabla P_\theta)(\xi).
\]
When $P_\theta$ admits a density $p_\theta$ with respect to a dominating measure and differentiation may again be interchanged with integration,
\[
\nabla J(\theta)
=
\int F(\xi)\nabla_\theta p_\theta(\xi)\,d\xi  
=
\int F(\xi)\frac{\nabla_\theta p_\theta(\xi)}{p_\theta(\xi)}
\,dP_\theta(\xi) 
=
\mathbb{E}\!\left[
F(\xi)\nabla_\theta\log p_\theta(\xi)
\right],
\]
which gives rise to the likelihood ratio (or score function) estimator
\[
\widehat{\nabla J}_{\rm LR}
=
\frac{1}{n}\sum_{i=1}^n
F(\xi_i)\nabla_\theta\log p_\theta(\xi_i).
\]

Although the preceding discussion focuses on first-order derivatives, the same principles extend naturally to higher-order derivatives through repeated differentiation under appropriate regularity conditions. Over the past five decades, numerous gradient estimation methods have been developed, most of which can be viewed as extensions of, or combinations of, these two fundamental ideas. For example, smoothed perturbation analysis and kernel-based methods extend the applicability of pathwise differentiation, while weak derivative and measure-valued differentiation methods generalize the likelihood ratio approach through alternative representations of probability measure derivatives. The generalized likelihood ratio method combines pathwise and likelihood ratio ideas to exploit the advantages of both. Consequently, differentiating either the performance function or the probability measure remains the prevailing conceptual framework for simulation gradient estimation.

\subsection{Limitations of Existing Approaches}

To illustrate the limitations of the existing approaches, we consider the classical problem of estimating the Delta of an Asian digital option under geometric Brownian motion. Under the risk-neutral measure, the asset price satisfies $dS_t=rS_t\,dt+\sigma S_t\,dW_t$ with initial value $S_0=\theta$, where $r$ is the risk-free rate and $\sigma>0$ is the volatility. Let $0<t_1<\cdots<t_m=T$ denote the monitoring times, and define the arithmetic average by $A_\theta=m^{-1}\sum_{j=1}^m S_{t_j}$. The discounted payoff is $F(A_\theta)=e^{-rT}\mathbf 1\{A_\theta\ge K\}$, and the option price is $J(\theta)=\mathbb E[F(A_\theta)]$. Our objective is to estimate the Delta $J'(\theta)$.

\subsubsection{Pathwise Methods and Nonsmooth Performance Functions.}

The pathwise method is not directly applicable because the discontinuity of the payoff invalidates the interchange of differentiation and expectation. Although the indicator payoff is differentiable with respect to \(\theta\) almost everywhere along each sample path, its derivative is zero almost surely. Consequently, the pathwise estimator is identically zero and would incorrectly imply that \(J'(\theta)=0\), even though the option price clearly depends on the initial asset price. The difficulty therefore lies not in the expected performance \(J(\theta)\), which is smooth, but in the discontinuity of the sample performance, which invalidates the pathwise argument.

Several extensions have been proposed to overcome this difficulty. Smoothed perturbation analysis restores differentiability by conditioning on carefully chosen random variables so that the conditional expectation becomes smooth \citep{fuhu1997,glasserman2004,fu2009conditional}. This approach requires exploiting the specific probabilistic structure of the simulation model and identifying suitable conditioning variables. Kernel-based methods replace the discontinuous payoff by a locally smoothed approximation and estimate the derivative of the smoothed performance \citep{liu2011}. While more generally applicable, they still require pathwise derivatives of the simulation output and typically incur a loss of statistical efficiency due to the additional smoothing.

\subsubsection{Likelihood-Ratio Methods and State-Space Enlargement.}

The likelihood-ratio method avoids differentiating the discontinuous payoff by differentiating the underlying probability measure instead. To explain its limitation, it is useful to first recall the change-of-measure representation. Fix a reference parameter $\theta_0$. If the probability measures $\{P_\theta\}$ are mutually absolutely continuous, then
\begin{equation}\label{eqn:IS}
J(\theta)
=
\mathbb E_{\theta_0}
\left[
F(\xi)L_{\theta,\theta_0}(\xi)
\right],
\end{equation}
where $L_{\theta,\theta_0}={dP_\theta}/{dP_{\theta_0}}$ is the likelihood ratio and \(\mathbb E_{\theta_0}\) is the expectation taken
with respect to \(P_{\theta_0}\). Differentiating this representation with respect to $\theta$ leads to the classical likelihood-ratio estimator. Therefore, the key ingredient of the likelihood-ratio method is an explicit expression for the likelihood ratio between the probability measures corresponding to different parameter values.

For the Asian digital option, the payoff depends on the asset-price process only through the arithmetic average $A_\theta$. Consequently, the most natural probability measure to differentiate is the distribution $P_\theta^A$ of $A_\theta$, under which
\[
J(\theta)
=
\int e^{-rT}\mathbf 1\{a\ge K\}\,dP_\theta^A(a).
\]
If the density of $A_\theta$ were available, one could directly construct the likelihood ratio $L_{\theta,\theta_0}^A={dP_\theta^A}/{dP_{\theta_0}^A}$, and obtain a likelihood-ratio estimator without differentiating the discontinuous payoff. Unfortunately, no tractable closed-form expression is known for the distribution of the arithmetic average of lognormal random variables. As a result, the likelihood ratio associated with the performance-relevant random variable cannot be evaluated.

The standard solution is to enlarge the state space from the arithmetic average $A_\theta$ to the discretized asset-price path $X_\theta=(S_{t_1},\ldots,S_{t_m})$, whose joint density is available from the transition densities of geometric Brownian motion \citep{glasserman2004}. The likelihood ratio is then constructed with respect to the enlarged probability measure $P_\theta^X$. By the Markov property,
\[
p_\theta(x_1,\ldots,x_m)
=
p_\theta^{(1)}(x_1)
\prod_{j=2}^m
p(x_j\mid x_{j-1}),
\]
where only the first transition density depends on the initial price $\theta$. Consequently,
\[
\left.
\frac{\partial}{\partial\theta}
\log p_\theta(X_\theta)
\right|_{\theta=\theta_0}
=
\frac{Z_1}
{\sigma\theta_0\sqrt{\Delta t}},
\]
where $\Delta t=t_1$ and $Z_1$ is the standard normal random variable driving the Brownian increment over the first time interval. Therefore, its second moment
\[
\mathbb E_{\theta_0}
\left[
\left(
\left.
\frac{\partial}{\partial\theta}
\log p_\theta(X_\theta)
\right|_{\theta=\theta_0}
\right)^2
\right]
=
\frac{1}{\sigma^2\theta_0^2\Delta t},
\]
which diverges as $\Delta t\rightarrow 0$, indicating an explosion of the variance of the likelihood-ratio gradient estimator as $\Delta t\rightarrow 0$.


The variance explosion admits a natural probabilistic interpretation. The distributions of the arithmetic average, \(P_\theta^A\) and \(P_{\theta_0}^A\), remain mutually absolutely continuous, so a likelihood ratio is well defined for the performance-relevant random variable. In contrast, when the likelihood ratio is constructed from the discretized sample path, the discretized state approaches the continuous sample path as \(\Delta t\to0\). Indeed, two geometric Brownian motions starting from different deterministic initial prices are mutually singular on the continuous path space, since every sample path starts from its initial value. Thus, the exploding likelihood-ratio score can be viewed as a finite-dimensional manifestation of this approaching singularity.

The Asian digital option illustrates a limitation that extends well beyond this example. In many stochastic simulation models, the probability distribution of the performance-relevant random variable is unavailable, making direct construction of the likelihood ratio impossible. Classical likelihood-ratio methods therefore rely on state-space enlargement to obtain an analytically tractable probability measure. 
While this restores applicability, it may substantially reduce statistical efficiency because the enlarged-state likelihood ratio retains variation beyond that summarized by the performance-relevant random variable.

More fundamentally, the likelihood-ratio method requires explicit knowledge of the probability density associated with the simulation model. In the Asian-option example, this density is available after enlarging the state space to the discretized asset-price path. In many modern simulation applications, however, the simulator is treated as a black box, and neither the probability distribution of the performance-relevant random variable nor that of any enlarged state space is available. Consequently, although the likelihood-ratio method is generally applicable to a broader class of performance functions than the pathwise method, it is also not applicable to black-box simulation models.

\section{Learn-Then-Differentiate Gradient Estimation}
\label{sec:ltd}

The methods reviewed in Section~\ref{sec:foundations} estimate derivatives by differentiating the expectation directly. An alternative strategy is to first construct a surrogate model for the expectation function and then obtain derivative estimates by differentiating the surrogate. From an engineering perspective, this is perhaps the most natural approach. Surrogate models are routinely employed to approximate computationally expensive simulation models for prediction, optimization, uncertainty quantification, and real-time decision making. Whenever the surrogate is sufficiently smooth, its derivatives become immediately available without requiring additional simulation experiments. Although this idea has been explored for several individual surrogate models, existing developments remain largely algorithm-specific \citep{hansen2008,neumeyer2010,liu2023derivatives}. In this section, we develop a unified framework for LTD gradient estimation.

The LTD framework is particularly attractive for modern simulation applications, where the simulator is often treated as a black box. In such settings, the analyst has access only to simulation inputs and their corresponding output observations,
\(
\mathcal D_n
=
\{(\theta_1,F_1),(\theta_2,F_2),\ldots,(\theta_n,F_n)\},
\)
where \(F_i\) denotes the simulation output observed at the design point \(\theta_i\). LTD requires neither pathwise derivatives nor analytical knowledge of the underlying probability model. Its objective is to estimate derivatives of the expectation function directly from these input-output observations.

In this paper, we focus on a broad class of surrogate models admitting the weighted representation
\begin{equation}
\widehat J_n(\theta)
=
\sum_{i=1}^n w_i(\theta)F_i,
\label{eq:weighted-surrogate}
\end{equation}
where the weights are determined by the dataset \(\mathcal D_n\) and the learning algorithm. We suppress their dependence on \(\mathcal D_n\) for notational simplicity.

This representation accommodates both response-independent and response-adaptive weights. When the weights depend only on the design points and fixed tuning parameters, the fitted surrogate is linear in the observed responses. Examples include linear regression, ridge regression, smoothing splines, kernel regression, local polynomial regression, kernel ridge regression, and regression using a fixed neural tangent kernel \citep{jacot2018ntk}. When tuning parameters are selected using the observed responses, these otherwise response-independent methods also fall within the response-adaptive case. More generally, response-adaptive weights also accommodate nonlinear learning procedures, including multiple kernel learning and certain neural-network estimators. Variants of this weighted representation have been widely used in the simulation learning literature \citep{ankenman2010,hongjiang2019,shen2021,lin2025}.

However, not every surrogate model of the form \eqref{eq:weighted-surrogate} is suitable for gradient estimation. For example, \(k\)-nearest-neighbor methods and regression trees also admit weighted representations but generally produce nondifferentiable approximations. Whenever the weight functions are \(r\) times differentiable, the corresponding LTD estimator is naturally defined by
\begin{equation}
\nabla^{(r)}\widehat J_n(\theta)
=
\sum_{i=1}^n
\nabla^{(r)}w_i(\theta)F_i.
\label{eq:ltd-general}
\end{equation}

An important motivation for constructing such surrogate models comes from the offline-simulation-online-application (OSOA) paradigm proposed by \citet{hongjiang2019}. In many modern simulation applications, the parameter values at which gradients are required are unavailable during the simulation stage, or the online stage has insufficient time or computational resources to perform new simulation experiments. Instead, computationally intensive simulation is performed offline to construct an accurate surrogate, while the online stage requires only evaluating the surrogate and its derivatives. Representative applications include contextual simulation optimization, where decisions and their sensitivities must be evaluated online for newly observed contextual information \citep{jin2026,lin2025}, and financial engineering, where option Greeks depend on market conditions that are only observed at the time of evaluation \citep{broadie1996,jiang2020}.

\subsection{Representative Weighted Surrogate Models}

The weighted representation \eqref{eq:weighted-surrogate} encompasses a broad class of differentiable nonparametric learning methods. In this subsection, we consider three representative surrogate models corresponding to three distinct learning paradigms: local averaging, local approximation, and global regularization. These examples illustrate the breadth of the unified LTD framework and serve as running examples throughout the paper. Their formulations, implementation details, and convergence properties are summarized in Appendix~\ref{app:surrogate} and will be used in the convergence analysis of Section~\ref{sec:ltd-rate}.

\paragraph{Kernel regression (KR).}

Kernel regression \citep{hansen2008} belongs to the class of \emph{local averaging} methods. The surrogate is constructed as a weighted average of nearby observations,
\[
\widehat J_n(\theta)
=
\frac{\sum_{i=1}^nK_h(\theta-\theta_i)F_i}
{\sum_{j=1}^nK_h(\theta-\theta_j)}
=
\sum_{i=1}^n
w_i^{\rm KR}(\theta)F_i,\qquad\text{where}\ 
\ w_i^{\rm KR}(\theta)
=
\frac{K_h(\theta-\theta_i)}
{\sum_{j=1}^nK_h(\theta-\theta_j)}.\]
The differentiability of the surrogate is inherited directly from the kernel function. The classical Nadaraya--Watson estimator
\citep{nadaraya1964,watson1964} employs a second-order kernel. More generally, higher-order kernels satisfying additional moment conditions can be used to achieve higher-order approximation accuracy. Such kernels generally take both positive and negative values, resulting in negative weights, but provide improved approximation accuracy for sufficiently smooth functions.

\vspace{6pt}

\paragraph{Local polynomial regression (LPR).}

Local polynomial regression \citep{fan1996} represents the philosophy of \emph{local approximation}. Rather than directly averaging nearby observations, it fits a polynomial of degree $p$ to observations in a neighborhood of the query point using kernel-weighted least squares. Although the estimator is defined implicitly through a local optimization problem, it admits the equivalent weighted representation
\[
\widehat J_n(\theta)
=
\sum_{i=1}^n
w_i^{\rm LP}(\theta)F_i,\qquad\text{where}\ \  w^{\rm LP}(\theta)^T
=
e_1^T
\big(R(\theta)^TW(\theta)R(\theta)\big)^{-1}
R(\theta)^TW(\theta).
\]
Here $R(\theta)$ is the local polynomial design matrix, $W(\theta)$ is the diagonal kernel-weight matrix, and $e_1=(1,0,\ldots,0)^T$. The weight functions therefore depend jointly on the localization kernel, the polynomial degree, the bandwidth, and the local design geometry. Since the polynomial basis is infinitely differentiable, the smoothness of the resulting surrogate is determined primarily by the localization kernel. 

\vspace{6pt}

\paragraph{Kernel ridge regression (KRR).}

Kernel ridge regression \citep{scholkopf2002} represents the philosophy of \emph{global regularization}. Rather than constructing a local approximation around each query point, it learns a global surrogate in a reproducing kernel Hilbert space (RKHS) by solving a regularized least-squares problem with regularization parameter $\lambda>0$. By the representer theorem,
\[
\widehat J_n(\theta)
=
k(\theta)^T(K+n\lambda I)^{-1}F
=
\sum_{i=1}^n
w_i^{\rm KRR}(\theta)F_i,\qquad\text{where}\ \ w^{\rm KRR}(\theta)^T
=
k(\theta)^T(K+n\lambda I)^{-1}.
\]
Consequently, the smoothness of the surrogate is inherited directly from the kernel function $k(\theta)$. Common choices include the Gaussian kernel and the Mat\'ern family, whose smoothness parameter explicitly controls the differentiability of the resulting surrogate.

\vspace{6pt}


Thus, the differentiability of all three learned surrogates above is governed primarily by the smoothness of the kernel function used in each estimator.
Whenever the learned surrogate is \(r\) times differentiable, it immediately gives rise to an LTD estimator of the \(r\)th-order derivative through \eqref{eq:ltd-general}. 
In contrast, the statistical performance is method-specific: as shown in Appendix~\ref{app:surrogate}, the approximation scheme determines the convergence rate of the learned surrogate, which is subsequently translated into the convergence rate of the corresponding LTD estimator in Section~\ref{sec:ltd-rate}.

\subsection{Two Questions}

The unified formulation developed in this section naturally raises two fundamental questions.

\vspace{5pt}

\noindent
\textit{Question 1. What mathematical object does an LTD estimator differentiate?}

Section~\ref{subsec:glasserman} reviewed Glasserman's unifying framework, in which simulation gradient estimators are interpreted as differentiating either the performance function or the underlying probability measure. Since LTD estimators are obtained by differentiating a learned surrogate, they do not fit naturally into this classical framework. Answering this question is therefore essential for understanding the fundamental nature of LTD and its relationship to the existing theory of simulation gradient estimation.

\vspace{5pt}

\noindent
\textit{Question 2. How efficient are LTD estimators?}

Ultimately, the practical value of an LTD estimator is determined by its statistical efficiency, which is typically measured by its convergence rate. The canonical Monte Carlo convergence rate of $n^{-1/2}$ has long been regarded as the benchmark for simulation gradient estimation \citep{glasserman2004,fu2006}. Meanwhile, for several weighted surrogate models, it is known that the surrogate itself can estimate the expectation function at a convergence rate close to $n^{-1/2}$ under suitable smoothness assumptions. A natural question is whether the corresponding LTD gradient estimators can inherit this near-optimal statistical efficiency. Answering this question is essential for understanding the theoretical performance of the LTD paradigm.

\vspace{5pt}

Sections~\ref{sec:ltd-interpretation} and~\ref{sec:ltd-rate} address these two questions, respectively.


\section{A Unified Interpretation of Learn-Then-Differentiate}
\label{sec:ltd-interpretation}

Section~\ref{sec:ltd} introduced a unified class of weighted LTD estimators of the form
\[
\nabla^{(r)}\widehat J_n(\theta)
=
\sum_{i=1}^n
\nabla^{(r)}w_i(\theta)F_i.
\]
We now return to the first question raised at the end of that section: what mathematical object does an LTD estimator differentiate?

At first glance, the answer appears immediate. Since \(\widehat J_n(\theta)\) is a learned surrogate for the expectation function, LTD seems to differentiate the surrogate itself. Viewed in this way, however, LTD does not fit naturally into Glasserman's framework introduced in Section~\ref{subsec:glasserman}, which interprets simulation gradient estimators as differentiating either the performance function or the underlying probability measure. The weighted representation above suggests a different interpretation.

To see this, recall that a likelihood-ratio estimator may be expressed as a weighted average of the observed performance values. For a first-order derivative, for example,
\[
\widehat{\nabla J}_{\mathrm{LR}}(\theta)
=
\frac{1}{n}
\sum_{i=1}^n
\nabla_\theta\log p_\theta(\xi_i)F(\xi_i).
\]
The likelihood-ratio weights arise from differentiating the probability measure, while the performance function \(F\) remains unchanged.

The weighted LTD estimator has a strikingly similar form: it combines observed performance values using the differentiated surrogate weights \(\nabla^{(r)}w_i(\theta)\). This structural similarity raises a natural question: can LTD also be interpreted as differentiating a representation of the probability measure while leaving the performance function fixed? Unlike the likelihood-ratio method, LTD does not begin with an analytical probability density. It begins with a surrogate learned from simulation observations. The central task is therefore to identify the measure representation associated with that surrogate and explain how differentiating it produces the LTD estimator.

\subsection{Probability Measures as Differentiable Objects}
\label{subsec:probability-measures}

To establish the proposed interpretation of LTD, we first construct a differentiable representation of the parameterized probability measure \(P_\theta\), through which expected performance and its derivatives may be expressed. Throughout this subsection, we focus on the population measure itself; learning its representation from simulation observations will be addressed in the next subsection. This parallels the population representation underlying the likelihood-ratio method in \eqref{eqn:IS}, where the likelihood ratio \(L_{\theta,\theta_0}\) represents \(P_\theta\) relative to a fixed reference measure \(P_{\theta_0}\), and differentiation acts on this representation.

\subsubsection{Kernel Mean Representation.}
Note that a probability density provides a familiar way to represent a
probability measure and calculate expectations:
\begin{equation}\label{eqn:density_rep}
J(\theta)
=
\int_{\mathcal X}F(\xi)p_\theta(\xi)\,d\xi.
\end{equation}
Here \(F\) describes the performance function being evaluated, while
\(p_\theta\) describes the distribution of the simulation outcome $\xi$, which takes values in \(\mathcal X\).
A kernel mean representation provides an alternative way to
express the same expectation, without requiring a probability
density.

To introduce this representation, choose a positive-definite kernel
\(k:\mathcal X\times\mathcal X\to\mathbb R\). The kernel determines
a reproducing kernel Hilbert space (RKHS), denoted by \(\mathcal H\).
An RKHS is a space of functions equipped with an inner product
\(\langle\cdot,\cdot\rangle_{\mathcal H}\). Its defining property is
that evaluating a function can be written as an inner product:
\[
F(\xi)=\langle F,k(\xi,\cdot)\rangle_{\mathcal H},
\qquad F\in\mathcal H.
\]
The condition \(F\in\mathcal H\) means that the performance function
belongs to the function class associated with the chosen kernel
and has finite RKHS norm. How restrictive this condition is
depends on both the kernel and the state space. For example,
Mat\'ern kernels on bounded domains accommodate broad classes
of sufficiently smooth functions. Continuous
kernels exclude discontinuous performance functions, and
uniformly bounded kernels exclude unbounded ones. We begin
with \(F\in\mathcal H\) to introduce the representation in its simplest form, and extend the discussion to more general piecewise smooth performance functions in Section~\ref{subsubsec:piecewise}.

For each parameter value \(\theta\), define the kernel mean representation of \(P_\theta\) by
\begin{equation}
\mu_\theta
=
\int_{\mathcal X}k(\xi,\cdot)\,dP_\theta(\xi)
=
\mathbb E_\theta[k(\xi,\cdot)].
\label{eq:kme}
\end{equation}
We assume the integrability condition
\begin{equation}
\mathbb E_\theta\left[\sqrt{k(\xi,\xi)}\right]<\infty,
\label{eq:kme-integrability}
\end{equation}
which ensures that this mean is well defined in \(\mathcal H\);
bounded kernels satisfy this condition.
Intuitively, \(\mu_\theta\) averages the kernel functions
\(k(\xi,\cdot)\) according to the probabilities assigned by
\(P_\theta\). It is a representation of the probability measure,
not a probability density, and is defined for discrete as well
as continuous distributions \citep{smola2007,muandet2017}.
The kernel here is defined on the simulation state space
\(\mathcal X\); it need not be the kernel used to fit a surrogate
on the parameter space \(\Theta\).

For \(F\in\mathcal H\), the reproducing property gives
\[
J(\theta)
=
\mathbb E_\theta[
\langle F,k(\xi,\cdot)\rangle_{\mathcal H}]
=
\langle F,\mu_\theta\rangle_{\mathcal H}.
\]
The interchange of expectation and the inner product is justified by linearity and continuity under the integrability condition~\eqref{eq:kme-integrability}.
This identity parallels the
density formula (\ref{eqn:density_rep}): integrating \(F\) against \(p_\theta\) is replaced
by taking the RKHS inner product of \(F\) with \(\mu_\theta\).
In particular, \(\mu_\theta\) determines the expectation of every
performance function in \(\mathcal H\).

This representation also makes the differentiation argument
transparent. For simplicity, consider a scalar parameter
\(\theta\). If \(\mu_\theta\) is \(r\)-times differentiable in
\(\mathcal H\), then, with \(F\) fixed, linearity and continuity of the inner
product give 
\begin{equation}
J^{(r)}(\theta)
=
\left\langle F,\mu_\theta^{(r)}\right\rangle_{\mathcal H},
\qquad
\mu_\theta^{(r)}
=
\frac{d^r\mu_\theta}{d\theta^r}.
\label{eq:kme-derivative}
\end{equation}
The same argument applies to partial derivatives for a vector
parameter.

As in differentiation through a probability density, the
performance function \(F\) remains fixed: differentiation acts on
the representation of the probability measure. The kernel mean representation therefore
provides a way to express measure differentiation without an
analytical density. This identity supplies the starting point
for interpreting the learned weighted surrogate and its
derivatives.

\subsubsection{Piecewise Smooth Performance Functions.}
\label{subsubsec:piecewise}

The preceding development assumes that \(F\in\mathcal H\).
Many simulation performance functions, however, are only
piecewise smooth. Examples include indicator functions,
positive-part functions, and piecewise-polynomial functions.
We now extend the representation to such functions, provided
that their smooth branches admit representations in the
chosen RKHS.

The idea is to divide the state space into regions and represent
the probability measure on each region separately. Suppose
that \(\mathcal X\) admits a finite partition
\(
\mathcal X=A_1\cup\cdots\cup A_M,
\)
independent of \(\theta\), and that
\begin{equation}
F(\xi)
=
\sum_{j=1}^M g_j(\xi)\mathbf 1_{\{\xi\in A_j\}},
\label{eq:piecewise-performance}
\end{equation}
where \(g_j\in\mathcal H\) agrees with the performance function
on \(A_j\). Its values outside \(A_j\) do not affect \(F\).
For each region, define the restricted measure
\[
P_{\theta,j}(B)=P_\theta(B\cap A_j).
\]
Its total mass is \(P_\theta(A_j)\), which may be less than one.
Thus it is a subprobability measure, rather than the conditional
distribution given \(A_j\): it retains both the probability of
entering the region and the distribution within it. Its kernel
mean embedding is
\begin{equation}
\mu_{\theta,j}
=
\int_{\mathcal X}k(\xi,\cdot)\,dP_{\theta,j}(\xi)
=
\mathbb E_\theta\!\left[
\mathbf 1_{\{\xi\in A_j\}}
k(\xi,\cdot)
\right].
\label{eq:restricted-kme}
\end{equation}
The integrability condition~\eqref{eq:kme-integrability}
also ensures that these embeddings are well defined. Then, the following theorem extends the expectation and derivative
representations to this setting.

\begin{theorem}
\label{thm:piece}
Suppose that \(F\) satisfies \eqref{eq:piecewise-performance},
with a parameter-independent partition and branch functions
\(g_j\in\mathcal H\). Under the integrability condition~\eqref{eq:kme-integrability},
\begin{equation}
J(\theta)
=
\sum_{j=1}^M
\langle g_j,\mu_{\theta,j}\rangle_{\mathcal H}.
\label{eq:piecewise-pop}
\end{equation}
Furthermore, if each mapping
\(\theta\mapsto\mu_{\theta,j}\) is \(r\)-times differentiable
in \(\mathcal H\), then, for a scalar parameter,
\begin{equation}
J^{(r)}(\theta)
=
\sum_{j=1}^M
\left\langle g_j,\mu_{\theta,j}^{(r)}
\right\rangle_{\mathcal H},\qquad\text{where}\ \ \mu_{\theta,j}^{(r)}
=
\frac{d^r\mu_{\theta,j}}{d\theta^r}.
\label{eq:piecewise-derivative}
\end{equation}
The same statement holds for partial derivatives of a vector
parameter.
\end{theorem}

\begin{proof}{Proof}
By \eqref{eq:piecewise-performance} and the reproducing property,
\[
J(\theta)
=
\sum_{j=1}^M
\mathbb E_\theta\!\left[
g_j(\xi)\mathbf 1_{\{\xi\in A_j\}}
\right]=
\sum_{j=1}^M
\left\langle
g_j,
\mathbb E_\theta\!\left[
\mathbf 1_{\{\xi\in A_j\}}k(\xi,\cdot)
\right]
\right\rangle_{\mathcal H}=
\sum_{j=1}^M\langle g_j,\mu_{\theta,j}\rangle_{\mathcal H}.
\]
The integrability condition~\eqref{eq:kme-integrability} justifies interchanging expectation
and the inner product. Since the branch functions and partition
are fixed, differentiating this finite sum gives (\ref{eq:piecewise-derivative}).
\hfill\Halmos
\end{proof}

Theorem~\ref{thm:piece} shows that extending the framework to piecewise smooth performance functions requires no change to the underlying differentiation principle. The only modification is that the original probability measure is replaced by finitely many subprobability measures, each associated with one smooth branch of the performance function. Differentiation continues to act exclusively on the representations of these probability measures while the performance functions remain fixed.

At first sight, introducing a partition of the state space appears to complicate the representation. This complication, however, is purely theoretical. As shown in the next subsection, after replacing the population embeddings by learned representations, the partition disappears completely from the resulting LTD estimator. The final estimator depends only on the observed performance values and the differentiated surrogate weights, regardless of whether the performance function is globally smooth or only piecewise smooth.

\subsection{Learning Probability Measure Representations}

Section~\ref{subsec:probability-measures} established that, at the population level,
expected performance can be expressed through kernel mean
representations of the parameterized probability measure.
Differentiation acts entirely on these representations while
leaving the performance function fixed. In practice, however,
the population representations are unavailable and must be
approximated from simulation observations. We now show that
differentiating their learned counterparts leads naturally to
the weighted LTD estimator introduced in Section~\ref{sec:ltd}.

Suppose simulation is performed at parameter values
\(\theta_1,\ldots,\theta_n\), yielding outcomes
\(\xi_i\sim P_{\theta_i}\) and performance observations
\(F_i=F(\xi_i)\), for \(i=1,\ldots,n\).
 To motivate learning the measure representations, consider the responses associated with each simulation outcome. 
 The performance response \(F_i\) has mean \(J(\theta_i)\), whereas, for each \(j\), the RKHS-valued response \(\mathbf 1_{\{\xi_i\in A_j\}}k(\xi_i,\cdot)\) has mean \(\mu_{\theta_i,j}\).
 Thus, learning \(\mu_{\theta,j}\) can be viewed as regression with RKHS-valued responses, paralleling regression for \(J(\theta)\) with performance observations. 
 This interpretation of conditional kernel mean estimation as regression is developed by \citet{grunewalder2012}.

Let \(w_i(\theta)\), \(i=1,\ldots,n\), denote the weights of the fitted surrogate. 
Motivated by this regression interpretation, we use the same fitted surrogate weights to combine the
RKHS-valued responses and define
\begin{equation}
\widehat\mu_{\theta,j}
=
\sum_{i=1}^n
w_i(\theta)
\mathbf 1_{\{\xi_i\in A_j\}}
k(\xi_i,\cdot),
\qquad j=1,\ldots,M.
\label{eq:restricted-kme-estimator}
\end{equation}
Equation~\eqref{eq:restricted-kme-estimator} has the same
weighted form as the surrogate models introduced in Section~\ref{sec:ltd}.
The difference lies in the object being approximated.

Differentiating the learned representations with respect
to $\theta$ gives
\[
\widehat\mu_{\theta,j}^{(r)}
=
\sum_{i=1}^n
w_i^{(r)}(\theta)
\mathbf 1_{\{\xi_i\in A_j\}}
k(\xi_i,\cdot),
\qquad j=1,\ldots,M.
\]
As before, we use scalar-parameter notation; the same argument
applies to partial derivatives for a vector parameter.
Replacing the population representations in
\eqref{eq:piecewise-derivative} by their learned counterparts
yields
\begin{equation}
\widehat J^{(r)}(\theta)
=
\sum_{j=1}^M
\left\langle
g_j,\widehat\mu_{\theta,j}^{(r)}
\right\rangle_{\mathcal H}.
\label{eq:kme-estimator}
\end{equation}
The following proposition shows that the estimator in
\eqref{eq:kme-estimator} is precisely the weighted LTD estimator
\eqref{eq:ltd-general} introduced in Section~\ref{sec:ltd}.

\begin{proposition}
\label{prop:piece}
Suppose that \(F\) satisfies \eqref{eq:piecewise-performance},
with \(g_j\in\mathcal H\), and that the fitted weights are
\(r\)-times differentiable. Then the estimator in
\eqref{eq:kme-estimator} satisfies
\[
\widehat J^{(r)}(\theta)
=
\sum_{i=1}^n w_i^{(r)}(\theta)F(\xi_i).
\]
\end{proposition}

\begin{proof}{Proof}
Using the definition of the learned kernel mean representations,
\[
\widehat J^{(r)}(\theta)
=
\sum_{j=1}^M\sum_{i=1}^n
w_i^{(r)}(\theta)
\mathbf 1_{\{\xi_i\in A_j\}}
\left\langle g_j,k(\xi_i,\cdot)\right\rangle_{\mathcal H}.
\]
Applying the reproducing property,
\(\langle g_j,k(\xi_i,\cdot)\rangle_{\mathcal H}=g_j(\xi_i)\),
gives
\[
\widehat J^{(r)}(\theta)
=
\sum_{j=1}^M
\sum_{i=1}^n
w_i^{(r)}(\theta)
\mathbf1\{\xi_i\in A_j\}
g_j(\xi_i)=
\sum_{i=1}^n w_i^{(r)}(\theta)
\left[
\sum_{j=1}^M
\mathbf 1_{\{\xi_i\in A_j\}}g_j(\xi_i)
\right]
=
\sum_{i=1}^n w_i^{(r)}(\theta)F(\xi_i),
\]
where the last equality follows from
\eqref{eq:piecewise-performance}.
\hfill\Halmos
\end{proof}

\begin{remark}{\it
Proposition~\ref{prop:piece} holds for both response-independent
and response-adaptive weights. To make this explicit, write
\(
w_i(\theta)=w_i(\theta;\mathcal D_n)\) where 
\(\mathcal D_n=\{(\theta_1,F_1),\ldots,(\theta_n,F_n)\}.
\)
For response-independent weights, the dependence on
\(\mathcal D_n\) is only through the design points; for
response-adaptive weights, it may also involve the observed
responses. In either case, the training data are held fixed
when differentiating:
\[
\frac{d^r}{d\theta^r}
\sum_{i=1}^n w_i(\theta;\mathcal D_n)F_i
=
\sum_{i=1}^n
\frac{\partial^r w_i(\theta;\mathcal D_n)}{\partial\theta^r}F_i.
\]
The proof of Proposition~\ref{prop:piece} therefore applies
unchanged, provided that the same fitted weights are used
across all regions.
 }
\end{remark}

Proposition~\ref{prop:piece} reveals an important feature of
the framework. Although the theoretical development uses a
partition to accommodate piecewise smooth performance
functions, this partition disappears from the final estimator.
The estimator depends only on the observed performance values
\(F(\xi_i)\) and the differentiated weights
\(w_i^{(r)}(\theta)\); neither the partition nor the branch
functions appear explicitly. Consequently, the partition
serves only as a theoretical device for interpreting nonsmooth
performance functions within the measure-differentiation
framework. It plays no role in implementing the resulting
LTD estimator.

\subsection{Relationship with Likelihood-Ratio Methods}

The development in the previous two subsections answers the first question raised at the end of Section~\ref{sec:ltd}. Although weighted LTD is obtained by differentiating a learned surrogate, the resulting estimator may equivalently be interpreted as differentiating a learned representation of the underlying probability measure. Consequently, weighted LTD fits naturally within Glasserman's unified framework for simulation gradient estimation as a probability-measure differentiation method. At the population level, the differentiation is expressed by
\[
J^{(r)}(\theta)
=
\left\langle
F,
\mu_\theta^{(r)}
\right\rangle_{\mathcal H}
\]
for globally smooth performance functions, and by its piecewise extension in Theorem~\ref{thm:piece}. At the empirical level, differentiating the learned kernel mean representations leads to the weighted estimator
\[
\widehat J^{(r)}(\theta)
=
\sum_{i=1}^n
w_i^{(r)}(\theta)F_i,
\]
which has exactly the same computational form as the unified LTD estimator~\eqref{eq:ltd-general} introduced in Section~\ref{sec:ltd}. Consequently, weighted LTD should be interpreted as differentiating a learned measure representation rather than the performance function itself.

The essential difference from the likelihood-ratio method lies in the representation being differentiated. Classical likelihood-ratio methods represent the probability measure through an analytical probability density or, more generally, a likelihood ratio. Differentiation therefore requires explicit knowledge of the underlying probability distribution. In contrast, the present framework represents the probability measure through its kernel mean representation and approximates this representation directly from simulation observations. No analytical probability density is required. Thus, while both approaches differentiate the probability measure, they differ fundamentally in how this probability measure is represented.

This distinction has an important consequence. In the likelihood-ratio method, the choice of the underlying random element is an algorithmic decision because it determines the probability density to be differentiated. If the probability distribution of the performance-relevant random variable is unavailable, one typically enlarges the state space to obtain an analytically tractable probability measure, as illustrated by the Asian digital option example in Section~\ref{sec:foundations}. Such state-space enlargement may substantially reduce statistical efficiency and, in continuous time, may even introduce singularity through the loss of absolute continuity.

The LTD framework is fundamentally different. Throughout Section~\ref{sec:ltd-interpretation}, the random element \(\xi\) serves only as a conceptual vehicle for constructing a representation of the probability measure. It may represent the performance-relevant random variable, the complete simulation trajectory, or any other measurable representation of the underlying randomness. Regardless of this choice, the final estimator depends only on the observed performance values and the differentiated surrogate weights. Consequently, the framework requires neither explicit knowledge of how the performance function depends on the underlying random element nor analytical knowledge of its probability distribution. Since no analytical density or likelihood ratio is differentiated, the singularity issues associated with state-space enlargement do not arise. This makes the proposed interpretation particularly well suited to modern black-box simulation models, where the underlying probability distribution is typically inaccessible.

\section{Statistical Efficiency of Learn-Then-Differentiate Estimators}
\label{sec:ltd-rate}

We now address the second question raised at the end of Section~\ref{sec:ltd}: how statistically efficient are LTD estimators? Statistical efficiency is typically measured by the convergence rate of an estimator as the simulation budget increases. Classical pathwise and likelihood-ratio estimators, when applicable, generally achieve the canonical Monte Carlo rate of \(n^{-1/2}\) \citep{glasserman2004,fu2006}. It is therefore natural to ask whether LTD estimators can attain, or at least approach, this benchmark.

For many learning algorithms, convergence rates of the learned surrogate \(\widehat J_n\) to the expectation function \(J\) are well understood \citep{fan1996,scholkopf2002,steinwart2008}. An LTD estimator, however, is obtained by differentiating the learned surrogate. Consequently, its convergence rate does not follow directly from existing function approximation results, since differentiation may amplify approximation errors. Existing analyses are largely method specific. In this section, we develop a unified convergence theory that directly transfers convergence rates of surrogate models to the corresponding LTD gradient estimators.

\subsection{Smoothness of the Expectation Function}
\label{subsec:expectation-smoothness}

To study higher-order derivatives of the expectation function \(J\),
we first examine its smoothness. Although the sample performance
function \(F\) may be nonsmooth or discontinuous, taking expectation
can average out these irregularities. Examples include simulation
models involving indicator functions, such as the Asian digital
option discussed in Section~\ref{sec:foundations}.

Throughout this section, the parameter space
\(\Theta\subset\mathbb R^d\) is a compact hyper-rectangle with nonempty
interior. For a multi-index \(\alpha=(\alpha_1,\ldots,\alpha_d)\),
where each component is a nonnegative integer, let
\(|\alpha|=\alpha_1+\cdots+\alpha_d\) and write
\[
D^\alpha f
=
\frac{\partial^{|\alpha|}f}
{\partial\theta_1^{\alpha_1}\cdots
 \partial\theta_d^{\alpha_d}}.
\]
We use \(D^{(r)}f=\{D^\alpha f:|\alpha|=r\}\) to denote the
collection of partial derivatives of order \(r\), and
\(C^\beta(\Theta)\) to denote the space of
\(\beta\)-times continuously differentiable functions on \(\Theta\).

The following proposition gives sufficient conditions for \(J\) to inherit smoothness from the density \(p_\theta\), even when \(F\) is nonsmooth.

\begin{proposition}
\label{prop:smooth}
Let \(\beta\ge1\) be an integer, and suppose
\[
J(\theta)
=
\int_{\mathcal X}F(\xi)p_\theta(\xi)\,d\xi,
\qquad \theta\in\Theta,
\]
where \(F\) is measurable and \(p_\theta\) is a probability density
on a fixed set \(\mathcal X\). Assume that \(p_\theta(\xi)\) is
\(\beta\)-times continuously differentiable in \(\theta\) for
almost every \(\xi\). Furthermore, for each \(|\alpha|\le\beta\),
suppose there is an integrable function \(G_\alpha\) such that
\begin{equation}\label{eq:integrable_bound}
|F(\xi)D_\theta^\alpha p_\theta(\xi)|
\le G_\alpha(\xi)
\end{equation} 
for all $\theta\in\Theta$ and almost every $\xi$. Then, \(J\) is \(\beta\)-times continuously differentiable, with
\[
D^\alpha J(\theta)
=
\int_{\mathcal X}F(\xi)D_\theta^\alpha p_\theta(\xi)\,d\xi,
\qquad |\alpha|\le\beta.
\]
\end{proposition}
\begin{proof}{Proof}
For each multi-index \(\alpha\) with \(|\alpha|\le\beta\), define
\[
H_\alpha(\theta)
=
\int_{\mathcal X}F(\xi)D_\theta^\alpha p_\theta(\xi)\,d\xi.
\]
The bound in \eqref{eq:integrable_bound} ensures that this integral
is finite. Continuity of \(D_\theta^\alpha p_\theta(\xi)\) and
dominated convergence imply that \(H_\alpha\) is continuous
on \(\Theta\).

\CompactDisplaySpacing
For \(|\alpha|<\beta\), let \(e_j\) be the \(j\)th coordinate
vector and take \(\theta\) in the interior of \(\Theta\).
By the mean value theorem and \eqref{eq:integrable_bound},
for sufficiently small nonzero \(h\),
\[
\left|
F(\xi)\frac{D_\theta^\alpha p_{\theta+he_j}(\xi)
             -D_\theta^\alpha p_\theta(\xi)}{h}
\right|
\le G_{\alpha+e_j}(\xi).
\]
Since the bound is integrable, dominated convergence gives
\[
\frac{\partial H_\alpha}{\partial\theta_j}(\theta)
=
\int_{\mathcal X}F(\xi)
D_\theta^{\alpha+e_j}p_\theta(\xi)\,d\xi
=
H_{\alpha+e_j}(\theta).
\]
Starting from \(H_0=J\), repeated differentiation establishes
the claimed formula through order \(\beta\). Continuity of
the functions \(H_\alpha\) on \(\Theta\) then yields
\(J\in C^\beta(\Theta)\).
\hfill\Halmos
\end{proof}

Proposition~\ref{prop:smooth} shows that the expectation function can inherit smoothness from the probability model even when the sample performance is discontinuous. If the assumptions of Proposition~\ref{prop:smooth} hold for every integer \(\beta\ge1\), then \(J\in C^\infty(\Theta)\). This observation motivates the use of smooth surrogates for estimating gradients and higher-order derivatives in simulation.

\subsection{A General Convergence Principle}
\label{subsec:general-ltd-rate}

Having examined the smoothness of the expectation function,
we now study how surface-fitting accuracy translates into
derivative-estimation accuracy. Let \(\widehat J_n\) be a
surrogate learned from \(n\) simulation observations. Our
objective is to establish convergence guarantees for the
LTD estimator \(D^{(r)}\widehat J_n\) from the corresponding
guarantees for \(\widehat J_n\).

\begingroup
\CompactDisplaySpacing
We consider two error criteria. The \(L_2\) norm measures
overall error across the parameter domain,
\[
\|f\|_{L_2(\Theta)}
=
\left(\int_\Theta |f(\theta)|^2\,d\theta\right)^{1/2},
\]
while the \(L_\infty\) norm measures the largest absolute error,
\[
\|f\|_{L_\infty(\Theta)}
=
\sup_{\theta\in\Theta}|f(\theta)|.
\]
An \(L_2\) guarantee therefore describes integrated accuracy,
whereas an \(L_\infty\) guarantee controls error uniformly
over the domain and, consequently, at every fixed parameter
value. For the collection \(D^{(r)}f\), define
\[
\|D^{(r)}f\|_{L_p(\Theta)}
=
\max_{|\alpha|=r}\|D^\alpha f\|_{L_p(\Theta)},
\qquad p\in\{2,\infty\}.
\]
Thus, the derivative error is measured by the largest
of the corresponding norms among all partial derivatives
of order \(r\).
\par\endgroup

\subsubsection{Interpolation Inequalities.}

%
%
%

The Gagliardo--Nirenberg interpolation inequalities connect the size of a function's derivatives to the size of the function itself and its higher-order smoothness. We use the following three special forms \citep{gagliardo1959,nirenberg1959}. 
The general inequality formulation are given in Appendix~\ref{app:gn}.

\begin{lemma}[Gagliardo--Nirenberg interpolation inequalities]
\label{lem:GN}
Let \(\Theta\subset\mathbb R^d\) be a compact hyperrectangle
with nonempty interior, and let \(r\) and \(\beta\) be integers
satisfying \(0\le r<\beta\). 
For every \(f\) for which the norms on the right-hand side below
are finite, the following inequalities hold:
\begin{eqnarray}
&&\|D^{(r)}f\|_{L_p(\Theta)}
\le C\left(
\|D^{(\beta)}f\|_{L_p(\Theta)}^{r/\beta}
\|f\|_{L_p(\Theta)}^{(\beta-r)/\beta}
+\|f\|_{L_p(\Theta)}\right),
\qquad p\in\{2,\infty\},
 \label{eq:GN_p}\\[8pt]
&&\|D^{(r)}f\|_{L_\infty(\Theta)}
\le C\left(
\|D^{(\beta)}f\|_{L_\infty(\Theta)}^{(2r+d)/(2\beta+d)}
\|f\|_{L_2(\Theta)}^{2(\beta-r)/(2\beta+d)}
+\|f\|_{L_2(\Theta)}\right),
 \label{eq:GN_2_inf}\\[8pt]
&& \|D^{(r)}f\|_{L_\infty(\Theta)}
\le C\left(
\|D^{(\beta)}f\|_{L_2(\Theta)}^{(r+d/2)/\beta}
\|f\|_{L_2(\Theta)}^{(\beta-r-d/2)/\beta}
+\|f\|_{L_2(\Theta)}\right),
\label{eq:GN_2_inf_L2}
\end{eqnarray}
where the last inequality holds provided that \(\beta>r+d/2\).
\end{lemma}

The first inequality controls the derivatives in the same
norm as the function, either \(L_2\) or \(L_\infty\).
The latter two control the derivatives in \(L_\infty\) using
the \(L_2\) norm of the function, with derivatives of order
\(\beta\) measured in \(L_\infty\) and \(L_2\), respectively.
Applied to the fitting error \(\widehat J_n-J\), these inequalities
connect surface-fitting accuracy to LTD estimation accuracy.

\subsubsection{Transferring Surface Rates to Derivative Rates.}

We first consider the case in which surface and derivative
errors are measured using the same criterion.

\begin{theorem}
\label{thm:ltd-general}
Let \(p\in\{2,\infty\}\), and let \(r\) and \(\beta\) be
integers satisfying \(0\le r<\beta\). Suppose that
\(J\in C^\beta(\Theta)\), and assume that
\[
\|D^{(\beta)}\widehat J_n\|_{L_p(\Theta)}=O_p(1),
\qquad
\|\widehat J_n-J\|_{L_p(\Theta)}=O_p(a_n),
\]
where \(a_n\to0\) is a positive deterministic sequence.
Then
\[
\|D^{(r)}\widehat J_n-D^{(r)}J\|_{L_p(\Theta)}
=
O_p\!\left(a_n^{(\beta-r)/\beta}\right).
\]
\end{theorem}

\begin{proof}{Proof}
Let \(e_n=\widehat J_n-J\). Since \(J\in C^\beta(\Theta)\),
\[
\|D^{(\beta)}e_n\|_{L_p(\Theta)}
\le
\|D^{(\beta)}\widehat J_n\|_{L_p(\Theta)}
+\|D^{(\beta)}J\|_{L_p(\Theta)}
=O_p(1).
\]
Also, \(\|e_n\|_{L_p(\Theta)}=O_p(a_n)\) by assumption.
Applying \eqref{eq:GN_p} to \(e_n\) gives
\[
\|D^{(r)}e_n\|_{L_p(\Theta)}
\le C\left(
\|D^{(\beta)}e_n\|_{L_p(\Theta)}^{r/\beta}
\|e_n\|_{L_p(\Theta)}^{(\beta-r)/\beta}
+\|e_n\|_{L_p(\Theta)}\right)=O_p\!\left(a_n^{(\beta-r)/\beta}+a_n\right)
=O_p\!\left(a_n^{(\beta-r)/\beta}\right).
\]
This concludes the proof.\hfill\Halmos
\end{proof}

\begin{remark}{\it
\citet{yatracos1989} established estimator-independent \(L_p\)-to-\(L_p\) derivative-rate transfer under H\"older-type regularity.
However, it is not generally clear how to verify the required regularity condition for estimators produced by modern machine learning methods.
Recent studies typically establish derivative convergence rates through method-specific analyses rather than directly applying Yatracos's transfer result \citep{fischer2020,liu2023derivatives}.

Instead, Theorem~\ref{thm:ltd-general} uses Lemma~\ref{lem:GN} under the condition \(\|D^{(\beta)}\widehat J_n\|_{L_p(\Theta)}=O_p(1)\).
This derivative-norm condition aligns naturally with the RKHS-norm controls available for KRR and multiple kernel learning (MKL); Appendix~\ref{app:derivative-bounds} provides method-specific justifications for the representative surrogate learning methods considered here.
Then, Theorem~\ref{thm:ltd-general} unifies existing derivative rates for KR, LPR, and KRR \citep{hansen2008,neumeyer2010,fischer2020} (Section~\ref{subsubsec:ltd-existing-rates}) and enables new guarantees for MKL and smooth neural networks (Section~\ref{subsubsec:ltd-further-rates}).}
\end{remark}


Theorem~\ref{thm:ltd-general} separates the learning problem
from the effect of differentiation. The learning method
determines the surface rate \(a_n\); once the derivative bound
\(\|D^{(\beta)}\widehat J_n\|_{L_p(\Theta)}=O_p(1)\) is established, the theorem gives the
derivative rate. For example, a surface rate
\(O_p(n^{-\gamma})\) yields a derivative rate
\(
O_p\!\left(n^{-\gamma(\beta-r)/\beta}\right).
\)
The factor \((\beta-r)/\beta\) describes the loss in the
guaranteed exponent associated with differentiation.
For a given surface rate, higher derivative orders incur
a greater loss, while greater smoothness reduces it.

In many applications, the available surface-fitting rate
is stated in the \(L_2\) norm, while the derivative is needed
at a particular parameter value. Inequalities~\eqref{eq:GN_2_inf}
and~\eqref{eq:GN_2_inf_L2} yield two \(L_\infty\) derivative
bounds, which also control the error at that point.

%

\begin{corollary}
\label{cor:l2-linf}
Let \(r\) and \(\beta\) be integers satisfying \(0\le r<\beta\).
Suppose that \(J\in C^\beta(\Theta)\) and
\[
\|\widehat J_n-J\|_{L_2(\Theta)}=O_p(a_n),
\]
where \(a_n\to0\) is a positive deterministic sequence.
Then the following statements hold.
\begin{enumerate}
\item[(i)] If \(\|D^{(\beta)}\widehat J_n\|_{L_\infty(\Theta)}=O_p(1)\), then
\[
\|D^{(r)}\widehat J_n-D^{(r)}J\|_{L_\infty(\Theta)}
=O_p\!\left(a_n^{2(\beta-r)/(2\beta+d)}\right).
\]
\item[(ii)] If \(\|D^{(\beta)}\widehat J_n\|_{L_2(\Theta)}=O_p(1)\) and \(\beta>r+d/2\), then
\[
\|D^{(r)}\widehat J_n-D^{(r)}J\|_{L_\infty(\Theta)}
=O_p\!\left(a_n^{(\beta-r-d/2)/\beta}\right).
\]
\end{enumerate}
In each case, the same rate holds for
\(|D^\alpha\widehat J_n(\theta_0)-D^\alpha J(\theta_0)|\)
at every fixed \(\theta_0\in\Theta\) and for every
multi-index \(\alpha\) with \(|\alpha|=r\).
\end{corollary}

\begin{proof}{Proof}
Let \(e_n=\widehat J_n-J\), so that
\(\|e_n\|_{L_2(\Theta)}=O_p(a_n)\).

In case (i), \(\|D^{(\beta)}\widehat J_n\|_{L_\infty(\Theta)}=O_p(1)\) and \(J\in C^\beta(\Theta)\) give \(\|D^{(\beta)}e_n\|_{L_\infty(\Theta)}=O_p(1).\)
Applying \eqref{eq:GN_2_inf} yields
\[
\|D^{(r)}e_n\|_{L_\infty(\Theta)}
=O_p\!\left(a_n^{2(\beta-r)/(2\beta+d)}+a_n\right)
=O_p\!\left(a_n^{2(\beta-r)/(2\beta+d)}\right).
\]

In case (ii), \(\|D^{(\beta)}\widehat J_n\|_{L_2(\Theta)}=O_p(1)\) and \(J\in C^\beta(\Theta)\) give \(\|D^{(\beta)}e_n\|_{L_2(\Theta)}=O_p(1).\)
Applying \eqref{eq:GN_2_inf_L2} yields
\[
\|D^{(r)}e_n\|_{L_\infty(\Theta)}
=O_p\!\left(a_n^{(\beta-r-d/2)/\beta}+a_n\right)
=O_p\!\left(a_n^{(\beta-r-d/2)/\beta}\right).
\]
The pointwise bounds follow from the uniform bounds.\hfill\Halmos
\end{proof}

\begin{remark}{\it
The weighted representation developed earlier explains
what a broad class of LTD estimators differentiates.
The convergence results above do not require that
representation: they depend only on the accuracy and
smoothness of the fitted surrogate. They therefore apply
to both response-independent and response-adaptive methods,
as well as to differentiable surrogates outside the
weighted framework, whenever the stated conditions hold.}
\end{remark}

\subsection{Representative LTD Estimators}
\label{subsec:ltd-examples}

\begingroup
\CompactDisplaySpacing
\setlength{\jot}{2pt}

We now illustrate the general convergence theory using five representative surrogate-learning methods. We first show how the transfer principle (i.e., Theorem~\ref{thm:ltd-general} and Corollary~\ref{cor:l2-linf}) recovers established derivative-rate exponents for KR, LPR, and KRR, whose weighted representations were introduced in Section~\ref{sec:ltd}. Next, we derive derivative convergence guarantees for multiple kernel learning and smooth neural networks fitted by least squares. 
These two examples extend the analysis to learning procedures that adapt the kernel combination or hidden features to the observed responses. 
Together, the five methods illustrate how surface-fitting rates and suitable smoothness control yield convergence guarantees for LTD estimators.

\subsubsection{Recovering Existing Convergence Results.}
\label{subsubsec:ltd-existing-rates}

We first show that our framework recovers derivative convergence rates established in the literature for KR, LPR, and KRR by applying the transfer principle to their surface-fitting rates.
Throughout this and the next subsection, we impose the smoothness and boundedness conditions required by the transfer principle.
We also assume that the design density is bounded above and away from zero on \(\Theta\), so that the distribution-weighted and Lebesgue \(L_2\) norms are equivalent.
Method-specific conditions for the surface-fitting rates are given in Appendix~\ref{app:surrogate}. Here \(d\) is the dimension of \(\theta\), \(n\) is the number of training observations, and \(0\le r<\beta\).

\paragraph{Kernel regression.}
KR with kernel order \(q\ge\beta\) and bandwidths chosen for the respective error criteria achieves the following surface-fitting rates under the conditions in Appendix~\ref{app:kr}:
 \begin{equation}\label{eqn:KR_rate}
 \|\widehat J_n-J\|_{L_2(\Theta)}
 =O_p\!\left(n^{-\beta/(2\beta+d)}\right),
 \qquad
 \|\widehat J_n-J\|_{L_\infty(\Theta)}
 =O_p\!\left[
 \left(\frac{\log n}{n}\right)^{\beta/(2\beta+d)}
 \right].
 \end{equation}
 
 At these respective bandwidths, \(\|D^{(\beta)}\widehat J_n\|_{L_2(\Theta)}=O_p(1)\) and \(\|D^{(\beta)}\widehat J_n\|_{L_\infty(\Theta)}=O_p(1)\) also hold (see Appendix~\ref{app:derivative-kr}).
Applying Theorem~\ref{thm:ltd-general} gives
\begin{align*}
\|D^{(r)}\widehat J_n-D^{(r)}J\|_{L_2(\Theta)}
&=O_p\!\left(n^{-(\beta-r)/(2\beta+d)}\right),\\
\|D^{(r)}\widehat J_n-D^{(r)}J\|_{L_\infty(\Theta)}
&=O_p\!\left[
\left(\frac{\log n}{n}\right)^{(\beta-r)/(2\beta+d)}
\right].
\end{align*}
The uniform derivative rate agrees with the rate established by \citet{hansen2008}.
Moreover, both the \(L_2\) and \(L_\infty\) derivative rates match the minimax rates established by \citet{stone1982optimal}.

\paragraph{Local polynomial regression.}
LPR with polynomial degree \(p\ge\beta-1\) and bandwidths chosen for the respective error criteria attains the same surface-fitting rates as \eqref{eqn:KR_rate} under the conditions in Appendix~\ref{app:lpr}.

At these respective bandwidths, \(\|D^{(\beta)}\widehat J_n\|_{L_2(\Theta)}=O_p(1)\) and \(\|D^{(\beta)}\widehat J_n\|_{L_\infty(\Theta)}=O_p(1)\) also hold (see Appendix~\ref{app:derivative-lpr}).
Applying Theorem~\ref{thm:ltd-general} gives
\begin{align*}
\|D^{(r)}\widehat J_n-D^{(r)}J\|_{L_2(\Theta)}
&=O_p\!\left(n^{-(\beta-r)/(2\beta+d)}\right),\\
\|D^{(r)}\widehat J_n-D^{(r)}J\|_{L_\infty(\Theta)}
&=O_p\!\left[
\left(\frac{\log n}{n}\right)^{(\beta-r)/(2\beta+d)}
\right].
\end{align*}
The uniform rate agrees with the derivative bounds of \citet[Lemma A.1]{neumeyer2010}.
 Again, both the \(L_2\) and \(L_\infty\) rates attain the  minimax rates of \citet{stone1982optimal}.
Note that here LTD differentiates the fitted intercept as a function of the query point, which need not coincide with the conventional derivative estimator obtained from the local polynomial coefficients \citep{fan1996,racine2016}.

\paragraph{Kernel ridge regression.}
For \(\beta>d/2\), consider KRR with a Mat\'ern kernel.
Suitable choices of kernel smoothness and regularization yield the following \(L_2\) surface-fitting rate under the conditions in Appendix~\ref{app:krr}:
\[
\|\widehat J_n-J\|_{L_2(\Theta)}
=O_p\!\left(n^{-\beta/(2\beta+d)}\right).
\]

With the same parameter choices, \(\|D^{(\beta)}\widehat J_n\|_{L_2(\Theta)}=O_p(1)\) also holds (see Appendix~\ref{app:derivative-krr}).
Applying Theorem~\ref{thm:ltd-general} gives the \(L_2\) rate below; when \(\beta>r+d/2\), Corollary~\ref{cor:l2-linf}(ii) also gives the \(L_\infty\) rate:
\begin{align*}
\|D^{(r)}\widehat J_n-D^{(r)}J\|_{L_2(\Theta)}
&=O_p\!\left(n^{-(\beta-r)/(2\beta+d)}\right),\\
\|D^{(r)}\widehat J_n-D^{(r)}J\|_{L_\infty(\Theta)}
&=O_p\!\left(n^{-(\beta-r-d/2)/(2\beta+d)}\right).
\end{align*}
The \(L_2\) derivative rate agrees with the rate established by \citet{fischer2020}.
For \(d=1\), \citet{liu2023derivatives} obtain sharper uniform surface-fitting and derivative rates under their specific assumptions.
 
\vspace{6pt}

These examples demonstrate that, under the stated surface-fitting and fitted-smoothness assumptions, a common surface-to-derivative argument recovers established convergence-rate exponents across distinct learning methods, including higher-order derivatives within the applicable smoothness ranges.

\subsubsection{Deriving Further Convergence Guarantees.}
\label{subsubsec:ltd-further-rates}

The framework also provides a route to derivative-rate guarantees that are not supplied by the available surface-fitting analyses. We illustrate this use with multiple kernel learning and smooth neural networks. Differentiating their fitted surfaces is a natural application of these learning methods; the purpose here is to establish convergence guarantees for the resulting estimators. Both examples can admit weighted representations with response-adaptive weights under the conditions described in Appendix~\ref{app:surrogate}. Their rate analysis follows the same surface-to-derivative argument used above.

\paragraph{Multiple kernel learning (MKL).}
 MKL constructs a surrogate by learning a combination of \(M\) candidate kernels \(k_m\), \(m\in\{1,\ldots,M\}\), together with the regression fit, allowing the kernel choice to adapt to the observed responses. We consider least-squares MKL with a nonnegative kernel combination whose weights sum to one. For the learned combination held fixed, the fitted surface has the KRR form and therefore admits response-adaptive observation weights; see Appendix~\ref{app:mkl}.

The candidate kernels depend on coordinate subsets of different dimensions \(1\le d_m\le d\). Let \(d_*:=\max_{1\le m\le M}d_m\).
For fixed \(M\), applying Theorem~1 of \citet{suzuki2011mkl} with suitably chosen candidate kernels and regularization yields the following \(L_2\) surface-fitting rate under the conditions in Appendix~\ref{app:mkl}:
\[
\|\widehat J_n-J\|_{L_2(\Theta)}
=O_p\!\left(n^{-\beta/(2\beta+d_*)}\right).
\]

With the same parameter choices, 
\(\|D^{(\beta)}\widehat J_n\|_{L_2(\Theta)}=O_p(1)\) also holds
(see Appendix~\ref{app:derivative-mkl}).
Applying Theorem~\ref{thm:ltd-general} gives the \(L_2\) rate below;
when \(\beta>r+d/2\), Corollary~\ref{cor:l2-linf}(ii) also gives the
\(L_\infty\) rate:
\begin{align*}
\|D^{(r)}\widehat J_n-D^{(r)}J\|_{L_2(\Theta)}
&=O_p\!\left(n^{-(\beta-r)/(2\beta+d_*)}\right),\\
\|D^{(r)}\widehat J_n-D^{(r)}J\|_{L_\infty(\Theta)}
&=O_p\!\left(n^{-(\beta-r-d/2)/(2\beta+d_*)}\right).
\end{align*}
 When \(d_*<d\), the component structure mitigates the curse of dimensionality; the uniform derivative bound still involves \(d\) through the interpolation step.

\paragraph{Smooth neural networks (NN).}
Consider a neural-network surrogate with a smooth nonpolynomial
activation, such as sigmoid or hyperbolic tangent (\(\tanh\)),
and a fixed number of hidden layers. Following the approximation
construction of \citet{park2026}, one hidden layer has width
of order \(N^d\), while the remaining hidden layers have
width at most two. The network weights are fitted by least
squares subject to a fixed bound on the fitted function and
its derivatives through order \(\beta\), where the bound is  chosen sufficiently large.
Under the additional output-layer conditions in
Appendix~\ref{app:nn-nonparametric}, the fitted network also
admits a response-adaptive weighted representation; those
conditions are not needed for the rate argument.

Park's theorem establishes an approximation error of order
\(N^{-\beta}\), rather than a statistical rate for a network fitted from observations. In
Proposition~\ref{prop:park-nn-surface} of
Appendix~\ref{app:nn-nonparametric}, we combine this
approximation result with a least-squares oracle inequality
\citep{koltchinskii2011} to establish
\[
\|\widehat J_n-J\|_{L_2(\Theta)}^2
=
O_p\!\left(
N^{-2\beta}+n^{-2\beta/(2\beta+d)}
\right).
\]
Choosing \(N\asymp n^{1/(2\beta+d)}\) gives
\[
\|\widehat J_n-J\|_{L_2(\Theta)}
=
O_p\!\left(n^{-\beta/(2\beta+d)}\right).
\]
The fitting constraint gives
\(\|D^{(\beta)}\widehat J_n\|_{L_2(\Theta)}=O(1)\) and \(\|D^{(\beta)}\widehat J_n\|_{L_\infty(\Theta)}=O(1)\).
Applying
Theorem~\ref{thm:ltd-general} and
Corollary~\ref{cor:l2-linf} (i) yields
\begin{align*}
\|D^{(r)}\widehat J_n-D^{(r)}J\|_{L_2(\Theta)}
&=
O_p\!\left(n^{-(\beta-r)/(2\beta+d)}\right),\\
\|D^{(r)}\widehat J_n-D^{(r)}J\|_{L_\infty(\Theta)}
&=
O_p\!\left(
n^{-2\beta(\beta-r)/(2\beta+d)^2}
\right).
\end{align*}
Thus, Park's approximation result supports a surface-fitting
guarantee for the trained network, which in turn yields
gradient and higher-order derivative guarantees through
the LTD framework.

\vspace{6pt}

Together, these two examples demonstrate the usefulness of the framework beyond recovering existing results. Once a surface-fitting rate and the required smoothness control are available, gradient and higher-order derivative guarantees follow without a separate statistical analysis for each derivative order.

\subsubsection{Smoothness, Dimensionality, and Statistical Efficiency.}

Under their respective assumptions, all five methods considered in Sections~\ref{subsubsec:ltd-existing-rates} and~\ref{subsubsec:ltd-further-rates} admit the common \(L_2\) derivative bound
\(
O_p\!\left(n^{-(\beta-r)/(2\beta+d)}\right).
\)
For MKL with \(d_*<d\), the component structure gives the faster bound displayed above.
The exponent
\(
{(\beta-r)}/{(2\beta+d)}
\)
describes the joint effects of smoothness, dimension, and derivative order: increasing \(d\) slows convergence, increasing \(\beta\) improves it, and increasing \(r\) reduces the rate. 

High smoothness is especially relevant in stochastic simulation. Proposition~\ref{prop:smooth} explains why a discontinuous sample performance function may have a smooth expectation. 
 When $J$ is sufficiently smooth, and the learning procedure supports the corresponding approximation and regularity conditions at arbitrarily large fixed smoothness orders, the displayed derivative-rate exponent approaches \(1/2\) for fixed \(r\) and \(d\), corresponding to the canonical Monte Carlo rate achieved by classical pathwise method and likelihood ratio method \citep{glasserman2004}.

The choices above, such as kernel order and bandwidth, local polynomial degree, kernel smoothness and regularization, and neural-network width and smoothness constraints, describe asymptotic regimes rather than prescriptions for implementation. In practice, target smoothness is unknown, and finite-sample performance depends on variance, numerical stability, and computational cost. Simpler models or response-based tuning may therefore be preferable. The unified framework provides a systematic way to translate surface-fitting guarantees into gradient and higher-order derivative guarantees whenever the required conditions hold, while retaining the practical advantages of LTD: learning from simulation outputs, reusing an offline surrogate, and evaluating derivatives throughout the parameter domain without further simulation.

\par
\endgroup

\section{Numerical Examples}
\label{sec:numerical}

We examine the empirical performance of LTD estimators through
two numerical examples with complementary roles. The first,
a portfolio of geometric Asian options, provides a benchmark
with closed-form prices and first- and second-order derivatives.
It allows direct assessment of estimation accuracy and
comparisons with standard pathwise estimators for first-order derivatives, the pathwise kernel estimators of \citet{liu2011} for second-order derivatives, and likelihood-ratio (LR) estimators for both orders.

The second example is a wireless communication network design
problem that illustrates derivative estimation under black-box access.
The parameters include transmit powers and antenna angles,
and derivative estimates must be constructed using only
simulation outputs. Under this black-box scenario,
sample-path derivatives and likelihood-ratio scores are
unavailable, so finite differences (FD) serve as the
classical comparator.

In both examples, we implement LTD using five surrogate-learning
methods: KR, LPR, KRR, MKL, and smooth
neural networks (NN). Each method estimates derivatives by
differentiating its fitted response surface. For each setting,
we use 50 independent replications and 100 fixed evaluation
points chosen independently of the training samples.
Reference derivatives are used only to assess accuracy.
Detailed model formulas and implementation procedures
are provided in Appendices~\ref{app:asian}
and~\ref{app:wireless}.

\subsection{Geometric Asian Option}
\label{subsec:numerical-asian}

We consider geometric Asian call options under risk-neutral geometric
Brownian motion, $dS_t=rS_t\,dt+\sigma S_t\,dW_t$, with initial price
$S_0=x$, risk-free rate $r$, volatility $\sigma$, and standard Brownian
motion $W_t$. For strike $K$, maturity $T$, and monitoring dates
$t_k=kT/m$, the discounted payoff is
\[
 F_m(x,\xi)=e^{-rT}\bigl(G_m(x)-K\bigr)^+,
 \qquad G_m(x)=\left(\prod_{k=1}^m S_{t_k}\right)^{1/m},
\]
where $(a)^+=\max\{a,0\}$. We set $K=100$, $r=0.04$, $\sigma=0.35$,
and $T=1$, and simulate exact GBM transitions. The price
$V_m(x)=\mathbb E[F_m(x,\xi)]$, Delta $V'_m(x)$, and Gamma $V''_m(x)$
have closed-form expressions (Appendix~\ref{app:asian}).

To examine the effect of input dimension, we form portfolios of \(d\in\{1,2,3,4\}\) otherwise identical options with independent asset paths, unit portfolio weights, and initial prices \(\mathbf{x}=(x_1,\ldots,x_d)^\top\). We estimate the \(d\) Deltas and \(d\) diagonal Gammas; mixed second derivatives are zero. We vary \(m\in\{50,200,1000\}\) because more frequent monitoring can increase the variance of the likelihood-ratio estimators considered here. This allows us to examine whether LTD, viewed through the measure-differentiation framework, exhibits similar sensitivity to monitoring frequency.

Our budgets reflect the distinction between online
estimation and offline learning. Pathwise and likelihood-ratio estimators each use
50,000 independent portfolio scenarios per evaluation point.
In applications such as dynamic risk management, that point
may become known only online, limiting the time available
for simulation. LTD instead fits a reusable surface offline
and evaluates derivatives at changing initial prices with
little online computation and no further simulation.
We therefore consider a range of offline budgets,
\(B\in\{50{,}000,500{,}000,5{,}000{,}000\}\), measured
in total scenarios per surface. The comparison examines
the benefits of offline investment rather than equal
total simulation costs.

At training site \(\mathbf{x}_i\), we average \(N_i\) simulated portfolio responses, with \(\sum_iN_i=B\), and fit KR, LPR, KRR, MKL, and NN surrogates to the resulting sample means. Training designs and site allocations are method-specific. Tuning uses the current replication's simulated responses: generalized cross-validation (GCV) for KR/LPR and NN initialization, robust GCV for KRR/MKL, and response validation for NN stopping (Appendix~\ref{app:asian}). Differentiating each fitted surface yields both Delta and Gamma estimates. Although the theoretical results give common convergence-rate exponents under the respective surface-fitting assumptions, finite-sample accuracy may differ across methods.

We use \(R=50\) independent replications and \(Q=100\)
fixed test points forming a Latin design in \([75,125]^d\),
independent of training samples and shared across methods
and monitoring frequencies.
Let \(g_{qj}\) be the exact derivative with respect to component \(j\)
at test point \(q\), and \(\widehat g_{bqj}\) its estimate
in replication \(b\). Define
\[
\operatorname{RMSE}_j
=\left\{\frac{1}{RQ}\sum_{b,q}(\widehat g_{bqj}-g_{qj})^2\right\}^{1/2},
\quad
\operatorname{rRMSE}_j
=\frac{\operatorname{RMSE}_j}{\{Q^{-1}\sum_q g_{qj}^2\}^{1/2}},
\quad
\operatorname{rRMSE}
=\frac{1}{d}\sum_{j=1}^d\operatorname{rRMSE}_j.
\]
Here $\operatorname{rRMSE}_j$ is the relative RMSE for component $j$,
and $\operatorname{rRMSE}$ is its mean across components.
Tables~\ref{tab:asian-delta-results}
and~\ref{tab:asian-gamma-results} report this mean as a percentage
over the $d$ Delta components and the $d$ diagonal-Gamma components,
respectively. Appendix~\ref{app:asian} provides implementation
details.

\begin{table}[t]
\centering\renewcommand{\baselinestretch}{1}\scriptsize
\renewcommand{\arraystretch}{0.9}\setlength{\tabcolsep}{4.5pt}
\caption{Delta rRMSE (\%) for the geometric Asian option example. 
}
\label{tab:asian-delta-results}
\begin{tabular}{rrrrrrrrrr}
\toprule
$d$ & $m$ & $B$ & PW & LR & KR & LPR & KRR & MKL & NN \\
\midrule
1 & 50 & 50k & 0.366 & 2.565 & 3.764 & 3.818 & 3.461 & 3.623 & \textbf{3.164} \\
1 & 50 & 500k &  &  & 1.660 & 1.232 & 1.434 & \textbf{1.225} & 3.937 \\
1 & 50 & 5m &  &  & 0.974 & 0.614 & 0.412 & \textbf{0.342} & 3.698 \\
\addlinespace[2pt]
1 & 200 & 50k & 0.350 & 6.458 & 3.614 & 3.919 & 3.042 & \textbf{2.909} & 3.291 \\
1 & 200 & 500k &  &  & 1.582 & 1.372 & 1.335 & \textbf{1.195} & 6.142 \\
1 & 200 & 5m &  &  & 0.943 & 0.620 & 0.427 & \textbf{0.355} & 0.584 \\
\addlinespace[2pt]
1 & 1000 & 50k & 0.381 & 12.177 & 3.854 & 3.477 & 2.863 & 2.794 & \textbf{2.756} \\
1 & 1000 & 500k &  &  & 1.717 & 1.356 & 1.317 & \textbf{1.161} & 2.343 \\
1 & 1000 & 5m &  &  & 0.876 & 0.604 & 0.410 & \textbf{0.333} & 0.635 \\
\addlinespace[2pt]
2 & 50 & 50k & 0.405 & 3.252 & 7.276 & 5.299 & 6.713 & 6.665 & \textbf{4.158} \\
2 & 50 & 500k &  &  & 3.981 & 2.468 & 2.581 & 2.249 & \textbf{1.730} \\
2 & 50 & 5m &  &  & 2.287 & 1.068 & 1.051 & 0.788 & \textbf{0.591} \\
\addlinespace[2pt]
2 & 200 & 50k & 0.391 & 5.566 & 7.680 & 5.093 & 6.966 & 7.030 & \textbf{4.635} \\
2 & 200 & 500k &  &  & 4.230 & 2.430 & 2.604 & 2.065 & \textbf{1.748} \\
2 & 200 & 5m &  &  & 2.387 & 1.058 & 1.099 & 0.845 & \textbf{0.638} \\
\addlinespace[2pt]
2 & 1000 & 50k & 0.391 & 11.633 & 7.347 & 5.761 & 6.002 & 6.422 & \textbf{4.038} \\
2 & 1000 & 500k &  &  & 4.014 & 2.423 & 2.586 & 1.953 & \textbf{1.604} \\
2 & 1000 & 5m &  &  & 2.466 & 1.090 & 1.088 & 0.787 & \textbf{0.577} \\
\addlinespace[2pt]
3 & 50 & 50k & 0.400 & 2.964 & 9.299 & 7.941 & 7.859 & 7.469 & \textbf{6.011} \\
3 & 50 & 500k &  &  & 7.047 & 6.723 & 3.074 & 3.083 & \textbf{2.518} \\
3 & 50 & 5m &  &  & 6.027 & 4.186 & 1.698 & 1.418 & \textbf{0.881} \\
\addlinespace[2pt]
3 & 200 & 50k & 0.382 & 5.796 & 9.467 & 8.110 & 8.245 & 7.856 & \textbf{6.290} \\
3 & 200 & 500k &  &  & 6.870 & 6.867 & 3.108 & 3.034 & \textbf{2.527} \\
3 & 200 & 5m &  &  & 6.127 & 4.388 & 1.754 & 1.459 & \textbf{0.864} \\
\addlinespace[2pt]
3 & 1000 & 50k & 0.410 & 11.032 & 9.309 & 8.158 & 8.273 & 7.668 & \textbf{6.141} \\
3 & 1000 & 500k &  &  & 7.240 & 6.996 & 3.216 & 3.083 & \textbf{2.650} \\
3 & 1000 & 5m &  &  & 6.444 & 4.563 & 1.672 & 1.425 & \textbf{0.842} \\
\addlinespace[2pt]
4 & 50 & 50k & 0.393 & 2.939 & 16.372 & 7.958 & 9.263 & 8.723 & \textbf{6.314} \\
4 & 50 & 500k &  &  & 13.774 & 7.967 & 6.638 & 5.764 & \textbf{3.168} \\
4 & 50 & 5m &  &  & 12.053 & 7.811 & 2.405 & 2.198 & \textbf{1.319} \\
\addlinespace[2pt]
4 & 200 & 50k & 0.433 & 5.542 & 16.449 & 8.049 & 9.750 & 9.294 & \textbf{6.714} \\
4 & 200 & 500k &  &  & 13.713 & 8.069 & 7.047 & 6.091 & \textbf{3.156} \\
4 & 200 & 5m &  &  & 12.398 & 7.962 & 2.479 & 2.266 & \textbf{1.315} \\
\addlinespace[2pt]
4 & 1000 & 50k & 0.402 & 11.746 & 17.121 & 8.105 & 9.612 & 9.203 & \textbf{6.765} \\
4 & 1000 & 500k &  &  & 14.044 & 8.062 & 6.762 & 6.147 & \textbf{3.174} \\
4 & 1000 & 5m &  &  & 12.203 & 8.116 & 2.475 & 2.255 & \textbf{1.310} \\
\bottomrule
\end{tabular}
\end{table}

\begin{table}[tbp]
\centering\renewcommand{\baselinestretch}{1}\scriptsize
\renewcommand{\arraystretch}{0.9}\setlength{\tabcolsep}{4.5pt}
\caption{Diagonal Gamma rRMSE (\%) for the geometric Asian option example.
}
\label{tab:asian-gamma-results}
\begin{tabular}{rrrrrrrrrr}
\toprule
$d$ & $m$ & $B$ & kernel PW & LR & KR & LPR & KRR & MKL & NN \\
\midrule
1 & 50 & 50k & 1.332 & 29.139 & 25.943 & \textbf{15.859} & 19.217 & 22.248 & 17.654 \\
1 & 50 & 500k &  &  & 15.499 & 8.882 & 7.347 & \textbf{5.365} & 11.477 \\
1 & 50 & 5m &  &  & 14.173 & 6.751 & 2.126 & \textbf{1.721} & 9.854 \\
\addlinespace[2pt]
1 & 200 & 50k & 1.224 & 100.641 & 22.541 & 14.718 & 12.493 & \textbf{11.806} & 16.239 \\
1 & 200 & 500k &  &  & 14.312 & 7.331 & 6.186 & \textbf{5.174} & 18.777 \\
1 & 200 & 5m &  &  & 12.719 & 5.685 & 2.229 & \textbf{1.752} & 3.529 \\
\addlinespace[2pt]
1 & 1000 & 50k & 1.325 & 437.127 & 24.574 & 12.938 & 11.214 & \textbf{10.235} & 14.263 \\
1 & 1000 & 500k &  &  & 15.662 & 9.158 & 6.104 & \textbf{5.067} & 13.484 \\
1 & 1000 & 5m &  &  & 12.006 & 5.347 & 2.089 & \textbf{1.605} & 4.017 \\
\addlinespace[2pt]
2 & 50 & 50k & 1.331 & 34.054 & 36.128 & 18.851 & 20.536 & 20.254 & \textbf{14.418} \\
2 & 50 & 500k &  &  & 25.453 & 11.274 & 11.296 & 8.594 & \textbf{6.930} \\
2 & 50 & 5m &  &  & 15.556 & 6.343 & 5.073 & 3.419 & \textbf{2.467} \\
\addlinespace[2pt]
2 & 200 & 50k & 1.276 & 101.165 & 36.588 & 18.889 & 21.168 & 21.119 & \textbf{16.589} \\
2 & 200 & 500k &  &  & 24.668 & 9.677 & 11.431 & 7.997 & \textbf{6.885} \\
2 & 200 & 5m &  &  & 15.401 & 6.110 & 5.171 & 3.566 & \textbf{2.603} \\
\addlinespace[2pt]
2 & 1000 & 50k & 1.309 & 473.820 & 35.597 & 20.273 & 18.159 & 19.029 & \textbf{13.757} \\
2 & 1000 & 500k &  &  & 25.336 & 10.559 & 11.370 & 7.547 & \textbf{5.885} \\
2 & 1000 & 5m &  &  & 14.862 & 6.297 & 5.190 & 3.318 & \textbf{2.365} \\
\addlinespace[2pt]
3 & 50 & 50k & 1.270 & 25.320 & 43.345 & 25.050 & 23.048 & 22.174 & \textbf{18.688} \\
3 & 50 & 500k &  &  & 35.322 & 21.381 & 12.114 & 11.319 & \textbf{9.518} \\
3 & 50 & 5m &  &  & 31.376 & 13.286 & 7.858 & 5.963 & \textbf{3.206} \\
\addlinespace[2pt]
3 & 200 & 50k & 1.280 & 100.614 & 43.474 & 25.381 & 24.093 & 23.234 & \textbf{19.985} \\
3 & 200 & 500k &  &  & 34.956 & 21.329 & 12.349 & 11.358 & \textbf{9.395} \\
3 & 200 & 5m &  &  & 32.449 & 14.399 & 8.109 & 6.169 & \textbf{3.371} \\
\addlinespace[2pt]
3 & 1000 & 50k & 1.335 & 452.039 & 43.451 & 25.302 & 24.017 & 22.769 & \textbf{19.434} \\
3 & 1000 & 500k &  &  & 35.705 & 21.626 & 12.559 & 11.294 & \textbf{9.797} \\
3 & 1000 & 5m &  &  & 31.793 & 14.925 & 7.788 & 6.031 & \textbf{3.235} \\
\addlinespace[2pt]
4 & 50 & 50k & 1.302 & 29.293 & 54.306 & 25.561 & 26.338 & 25.052 & \textbf{20.725} \\
4 & 50 & 500k &  &  & 50.200 & 24.971 & 21.513 & 20.544 & \textbf{10.967} \\
4 & 50 & 5m &  &  & 48.512 & 24.410 & 10.522 & 8.919 & \textbf{4.842} \\
\addlinespace[2pt]
4 & 200 & 50k & 1.268 & 103.854 & 53.764 & 25.555 & 27.246 & 26.069 & \textbf{21.363} \\
4 & 200 & 500k &  &  & 49.759 & 25.248 & 22.436 & 21.177 & \textbf{11.295} \\
4 & 200 & 5m &  &  & 48.029 & 24.661 & 10.724 & 9.082 & \textbf{4.869} \\
\addlinespace[2pt]
4 & 1000 & 50k & 1.338 & 534.658 & 53.700 & 25.669 & 26.968 & 25.927 & \textbf{21.782} \\
4 & 1000 & 500k &  &  & 50.083 & 25.128 & 21.663 & 21.239 & \textbf{11.118} \\
4 & 1000 & 5m &  &  & 47.941 & 25.081 & 10.741 & 9.103 & \textbf{4.857} \\
\bottomrule
\end{tabular}
\end{table}

As monitoring becomes more frequent, LR errors increase substantially, particularly for Gamma, whereas LTD errors generally show much less sensitivity to monitoring frequency. Across all methods, including the pathwise and LR estimators, relative errors are higher for Gamma than for Delta, reflecting the greater difficulty of estimating second-order derivatives. LTD errors also generally increase with dimension, with KR appearing particularly sensitive to both dimension and derivative order. Among the LTD methods, MKL generally performs best in one dimension, especially at the larger budgets, while NN achieves the lowest Delta and Gamma errors in every reported setting for $d=2,3,4$. These differences illustrate that common asymptotic rate exponents do not imply equal finite-sample accuracy. The pathwise estimators are more accurate than all five LTD methods at the smallest LTD budget, while larger offline budgets generally improve LTD accuracy. Once fitted, each LTD surface provides both Delta and Gamma estimates at new initial prices without further simulation, a feature useful for repeated sensitivity evaluation in applications such as dynamic risk management.

\subsection{Wireless Communication Network}
\label{subsec:numerical-wireless}

Wireless network performance depends on how transmit powers and antenna
orientations balance coverage and interference. Performance derivatives
can guide simulation optimization and support real-time control as
operating conditions change. A measure of communication quality
is the signal-to-interference-plus-noise ratio (SINR), which compares
received signal power with interference and noise. We use expected
log-SINR as the objective, averaging communication quality on a
logarithmic scale. In practice, network performance is often evaluated
using complex simulators that incorporate network geometry and signal
propagation, with internal calculations inaccessible to the analyst.
For illustration, we construct a semi-analytical model based on
standard propagation mechanisms, as described in
Appendix~\ref{app:wireless}. Its tractability facilitates the computation
of accurate reference gradients, while the estimators use only
black-box simulation outputs.

The model consists of two directional antennas with design vector
$\theta=(p_1,\alpha_1,\beta_1,p_2,\alpha_2,\beta_2)^\top$,
where $p_i$, $\alpha_i$, and $\beta_i$ denote the transmit power, azimuth angle, and downtilt angle of antenna $i$, respectively.
For each design, the simulator returns a log-SINR observation
$F(\theta,\xi)$, where $\xi$ represents random user location and
shadowing. We estimate $\nabla J(\theta)$, where
$J(\theta)=\mathbb{E}[F(\theta,\xi)]$, for three sets of active
variables: the two powers, the four angles, and all six coordinates,
corresponding to $d=2,4,6$, respectively.

Under the assumed black-box access, sample-path derivatives and
likelihood-ratio scores are unavailable, leaving FD as the applicable
benchmark among the classical methods considered here. We use
central FD,
$\widehat g_j^{\mathrm{FD}}(\theta)
=[\overline F(\theta+h_je_j)-\overline F(\theta-h_je_j)]/(2h_j)$,
where $e_j$ is the coordinate unit vector, $h_j$ is the perturbation
half-step, and $\overline F$ is a sample mean. Each gradient evaluation
uses approximately $10^5$ observations allocated equally among the
$2d$ perturbed settings, with independent samples at the positive and
negative perturbations. Following the offline--online comparison in
Section~\ref{subsec:numerical-asian}, LTD uses
$B\in\{10^5,10^6,10^7\}$ observations per fitted surface, allowing
greater offline investment to support repeated online evaluations.
We report results for KR, LPR, KRR, and MKL with GCV-selected
tuning parameters, and for NN with GCV-selected initialization
parameters and a stopping stage chosen by response validation
(Appendix~\ref{app:wireless}).

We conduct $R=50$ independent replications and evaluate each method
at $Q=100$ fixed test points from a centered Latin hypercube design
\citep{mckay1979lhs}, inset
from the design boundary to permit central differences. These points
are independent of the training samples and shared across methods;
each LTD surface serves all test points. Reference gradients are
computed separately by central differences using $10^6$ observations
at each perturbed setting and common random numbers (CRNs) across
paired perturbations. CRNs require control over simulation randomness,
which is unavailable under the black-box access assumed for the
estimators. The reference gradients are used only to assess accuracy.

We report two complementary accuracy measures. To aggregate errors
across power and angle coordinates, we use the range-normalized RMSE,
\[
 \operatorname{nRMSE}
 =
 \left\{
 \frac{R^{-1}\sum_{b=1}^R\sum_{q=1}^Q\sum_{j=1}^d
 s_j^2(\widehat g_{bqj}-g_{qj})^2}
 {\sum_{q=1}^Q\sum_{j=1}^d s_j^2g_{qj}^2}
 \right\}^{1/2},
\]
where $g_{qj}$ is the reference derivative for coordinate $j$ at
test point $q$, $\widehat g_{bqj}$ is its estimate in replication
$b$, and $s_j$ is half the design range of coordinate $j$.
Scaling by $s_j$ makes this aggregate measure invariant to the
units of the design variables. To examine individual derivatives,
we also report the componentwise relative RMSE $\operatorname{rRMSE}_j$
defined in Section~\ref{subsec:numerical-asian}, using the numerical
reference gradients $g_{qj}$.
Tables~\ref{tab:wireless-overall-nrmse}
and~\ref{tab:wireless-component-relrmse}
report these overall and componentwise measures.

\begin{table}[tbp]
\centering\renewcommand{\baselinestretch}{1}\scriptsize
\renewcommand{\arraystretch}{0.9}\setlength{\tabcolsep}{4.5pt}
\caption{Overall nRMSE (\%) for the wireless communication network example.
}
\label{tab:wireless-overall-nrmse}
\begin{tabular}{rrrrrrrr}
\toprule
$d$ & $B$ & FD & KR & LPR & KRR & MKL & NN \\
\midrule
2 & 100k & 11.162 & 24.791 & \textbf{8.203} & 8.605 & 9.416 & 13.775 \\
2 & 1m &  & 18.682 & \textbf{2.973} & 3.138 & 3.118 & 4.737 \\
2 & 10m &  & 17.827 & 1.309 & \textbf{1.262} & 1.369 & 1.938 \\
\addlinespace[2pt]
4 & 100k & 10.747 & 30.005 & 16.605 & 13.069 & \textbf{13.031} & 15.419 \\
4 & 1m &  & 26.928 & 10.703 & \textbf{7.540} & 7.696 & 9.055 \\
4 & 10m &  & 25.574 & 7.692 & \textbf{3.946} & 4.412 & 5.444 \\
\addlinespace[2pt]
6 & 100k & 13.826 & 35.277 & 21.747 & 16.021 & \textbf{15.971} & 19.433 \\
6 & 1m &  & 32.810 & 14.841 & 9.024 & \textbf{8.652} & 13.251 \\
6 & 10m &  & 32.905 & 12.071 & 5.235 & \textbf{5.113} & 13.826 \\
\bottomrule
\end{tabular}
\end{table}

The overall errors of all five LTD estimators increase across the
$d=2,4,6$ settings. KR consistently has the largest overall error,
while NN is less competitive than in the Asian option example.
For $d=2$, LPR and KRR perform similarly, with LPR achieving slightly
lower overall errors at the small and moderate budgets and KRR at
the largest budget. For $d=4,6$, KRR and MKL are the most accurate
LTD methods overall: KRR has a modest advantage at the larger
budgets for $d=4$, while MKL has slightly lower overall errors at
every budget for $d=6$. These results show that the relative
performance of LTD methods depends on the application as well as
the dimension and simulation budget.

At $B=10^5$, matching the simulation budget of one FD gradient evaluation, the best LTD methods achieve lower overall errors than FD for $d=2$, but moderately higher errors for $d=4,6$. At $B=10^6$ and $10^7$, both KRR and MKL outperform the fixed-budget FD comparator in overall nRMSE and every componentwise rRMSE across all three settings. Notice that, different from FD estimators, each LTD surface also supplies gradients at new designs without further simulation. Together, these results illustrate how greater offline investment can
improve derivative accuracy while supporting repeated online evaluations in black-box simulation optimization and real-time control.

\begin{table}[tbp]
\centering\renewcommand{\baselinestretch}{1}\scriptsize
\renewcommand{\arraystretch}{0.9}\setlength{\tabcolsep}{4.5pt}
\caption{Componentwise rRMSE (\%) for the wireless communication network example.
}
\label{tab:wireless-component-relrmse}
\begin{tabular}{rrrrrrrrr}
\toprule
$d$ & Parameter & $B$ & FD & KR & LPR & KRR & MKL & NN \\
\midrule
2 & $p_1$ & 100k & 11.956 & 24.330 & \textbf{9.003} & 9.419 & 10.673 & 14.875 \\
2 & $p_1$ & 1m &  & 19.361 & \textbf{3.425} & 3.681 & 3.585 & 5.072 \\
2 & $p_1$ & 10m &  & 17.153 & \textbf{1.381} & 1.437 & 1.565 & 2.141 \\
\addlinespace[2pt]
2 & $p_2$ & 100k & 10.508 & 25.142 & \textbf{7.528} & 7.920 & 8.317 & 12.861 \\
2 & $p_2$ & 1m &  & 18.140 & \textbf{2.569} & 2.644 & 2.703 & 4.462 \\
2 & $p_2$ & 10m &  & 18.331 & 1.250 & \textbf{1.107} & 1.195 & 1.767 \\
\addlinespace[2pt]
4 & $\alpha_1$ & 100k & 39.484 & 88.744 & 40.979 & \textbf{40.304} & 40.548 & 50.417 \\
4 & $\alpha_1$ & 1m &  & 86.182 & 23.704 & \textbf{19.308} & 19.986 & 30.286 \\
4 & $\alpha_1$ & 10m &  & 83.755 & 15.389 & \textbf{10.064} & 10.934 & 18.398 \\
\addlinespace[2pt]
4 & $\beta_1$ & 100k & 7.869 & 23.681 & 13.282 & 10.503 & \textbf{10.468} & 13.096 \\
4 & $\beta_1$ & 1m &  & 20.112 & 9.133 & \textbf{5.830} & 6.013 & 6.810 \\
4 & $\beta_1$ & 10m &  & 18.154 & 6.612 & \textbf{2.791} & 3.198 & 3.778 \\
\addlinespace[2pt]
4 & $\alpha_2$ & 100k & 37.796 & 85.875 & 39.623 & \textbf{36.826} & 37.525 & 49.373 \\
4 & $\alpha_2$ & 1m &  & 81.997 & 21.382 & \textbf{19.294} & 19.824 & 29.786 \\
4 & $\alpha_2$ & 10m &  & 80.425 & 15.528 & \textbf{9.885} & 10.651 & 17.041 \\
\addlinespace[2pt]
4 & $\beta_2$ & 100k & 10.199 & 30.924 & 18.215 & 13.159 & \textbf{13.015} & 13.724 \\
4 & $\beta_2$ & 1m &  & 27.898 & 11.575 & \textbf{8.330} & 8.408 & 8.972 \\
4 & $\beta_2$ & 10m &  & 27.008 & 8.353 & \textbf{4.620} & 5.148 & 5.785 \\
\addlinespace[2pt]
6 & $p_1$ & 100k & 21.070 & 47.828 & 20.759 & 17.854 & \textbf{17.650} & 28.009 \\
6 & $p_1$ & 1m &  & 38.405 & 13.233 & 9.628 & \textbf{9.355} & 17.749 \\
6 & $p_1$ & 10m &  & 40.973 & 11.153 & 5.549 & \textbf{5.066} & 14.756 \\
\addlinespace[2pt]
6 & $\alpha_1$ & 100k & 48.665 & 93.147 & 54.763 & \textbf{50.740} & 51.174 & 52.131 \\
6 & $\alpha_1$ & 1m &  & 93.133 & 39.762 & \textbf{24.431} & 24.831 & 33.188 \\
6 & $\alpha_1$ & 10m &  & 91.031 & 33.485 & 15.321 & \textbf{14.804} & 26.950 \\
\addlinespace[2pt]
6 & $\beta_1$ & 100k & 10.486 & 31.590 & 18.989 & 14.154 & \textbf{13.794} & 16.972 \\
6 & $\beta_1$ & 1m &  & 28.292 & 12.978 & 7.749 & \textbf{7.342} & 11.254 \\
6 & $\beta_1$ & 10m &  & 29.084 & 10.625 & 4.820 & \textbf{4.757} & 13.205 \\
\addlinespace[2pt]
6 & $p_2$ & 100k & 20.769 & 45.076 & 20.736 & 16.870 & \textbf{16.664} & 26.545 \\
6 & $p_2$ & 1m &  & 37.904 & 13.512 & 9.566 & \textbf{9.232} & 16.725 \\
6 & $p_2$ & 10m &  & 37.485 & 11.192 & 5.449 & \textbf{5.025} & 14.885 \\
\addlinespace[2pt]
6 & $\alpha_2$ & 100k & 52.044 & 100.237 & 55.753 & \textbf{53.008} & 54.435 & 57.026 \\
6 & $\alpha_2$ & 1m &  & 93.902 & 38.595 & 26.604 & \textbf{26.155} & 35.699 \\
6 & $\alpha_2$ & 10m &  & 93.036 & 33.963 & 15.074 & \textbf{14.417} & 26.742 \\
\addlinespace[2pt]
6 & $\beta_2$ & 100k & 10.647 & 30.505 & 21.643 & \textbf{13.720} & 13.832 & 16.536 \\
6 & $\beta_2$ & 1m &  & 30.189 & 14.703 & 8.569 & \textbf{8.124} & 12.359 \\
6 & $\beta_2$ & 10m &  & 29.636 & 11.608 & 4.614 & \textbf{4.561} & 13.306 \\
\bottomrule
\end{tabular}
\end{table}


\section{Conclusion}
\label{sec:conclusion}

This paper develops a unified framework for understanding
and analyzing LTD estimators.
For surrogates admitting a weighted representation, LTD
differentiates a learned representation of the underlying
probability measure while leaving the performance function
fixed. A general convergence theorem then translates
surface-fitting accuracy into guarantees for gradients and
higher-order derivatives. Applications recover existing
results for representative learning methods and establish
additional guarantees for multiple kernel learning and
smooth neural networks. Numerical examples demonstrate
the value of offline learning for accurate, repeated
derivative evaluations without further simulation,
particularly in black-box applications.

\bibliographystyle{informs2014}
\bibliography{references}

\ECSwitch
\ECHead{Electronic Companion}

\normalsize

\renewcommand{\theassumption}{\Alph{section}.\arabic{assumption}}
\setcounter{assumption}{0}

\renewcommand{\thelemma}{\Alph{section}.\arabic{lemma}}
\setcounter{lemma}{0}

\renewcommand{\theproposition}{\Alph{section}.\arabic{proposition}}
\setcounter{proposition}{0}

\renewcommand{\thefigure}{\Alph{section}.\arabic{figure}}
\setcounter{figure}{0}

\renewcommand{\thetable}{\Alph{section}.\arabic{table}}
\setcounter{table}{0}


\section{Representative Surrogate Learning Methods}
\label{app:surrogate}

This appendix describes the five surrogate-learning methods used in Section~\ref{subsec:ltd-examples}: kernel regression, local polynomial regression, kernel ridge regression, multiple kernel learning, and smooth neural networks. For each method, we present the fitting procedure, its weighted representation, the principal tuning choices, and the surface-fitting convergence results underlying the LTD analysis.
In particular, we derive a statistical surface-fitting rate for the neural-network method by combining Park's approximation theorem with a least-squares oracle inequality.
These surface-fitting rates are combined with the required higher-order derivative bounds to obtain the corresponding LTD convergence rates in Section~\ref{subsec:ltd-examples}.

The first three methods have response-independent weights when their tuning parameters are fixed in advance, and response-adaptive weights when these parameters are selected using the observed responses, for example through cross-validation.  
Multiple kernel learning learns its kernel combination from the responses, while jointly trained neural networks can also admit response-adaptive weights under the output-layer conditions stated below. 
In all cases, we hold the training data and the tuning parameters fixed when differentiating the fitted surrogate with respect to the query point.

\subsection{Kernel Regression}\label{app:kr}

Kernel regression, also known as the Nadaraya--Watson estimator, approximates the expectation function through local averaging \citep{nadaraya1964,watson1964}. Given design points \(\theta_1,\ldots,\theta_n\) and responses \(F_1,\ldots,F_n\), the estimator is
\[
\widehat J_n(\theta)
=\frac{\sum_{i=1}^n K_h(\theta-\theta_i)F_i}
{\sum_{i=1}^n K_h(\theta-\theta_i)}
=\sum_{i=1}^n w_i(\theta)F_i,
\qquad
w_i(\theta)
=\frac{K_h(\theta-\theta_i)}
{\sum_{j=1}^n K_h(\theta-\theta_j)},
\]
where \(K_h(u)=h^{-d}K(u/h)\) and \(h>0\) is the bandwidth. For a nonnegative kernel, the weights form a local average of nearby responses. Higher-order kernels may take negative values, but the weights still sum to one wherever the denominator is nonzero. Thus, both choices admit the weighted representation introduced in Section~\ref{sec:ltd}.

The bandwidth controls the bias--variance tradeoff, while the kernel order determines the extent of bias reduction. A kernel of order \(q\) satisfies
\[
\int_{\mathbb R^d}K(u)\,du=1,
\qquad
\int_{\mathbb R^d}u^\alpha K(u)\,du=0,
\quad 1\le |\alpha|<q,
\]
with a nonzero moment for at least one multi-index of order \(q\), and appropriate absolute moments finite. Kernel order is distinct from differentiability: a kernel used for LTD must also be sufficiently smooth so that the fitted surrogate can be differentiated to the desired order. These derivatives are well defined on regions where the denominator stays away from zero.

Suppose \(J\in C^\beta(\Theta)\), and let
\(
s_K=\min\{\beta,q\}.
\)
Under standard random-design regression conditions, including a sufficiently smooth design density bounded away from zero and suitable moment and tail conditions on the centered response noise, the surface-fitting errors satisfy
\begin{align*}
\|\widehat J_n-J\|_{L_2(\Theta)}
&=O_p\!\left(h^{s_K}+(nh^d)^{-1/2}\right),\\
\|\widehat J_n-J\|_{L_\infty(\Theta)}
&=O_p\!\left(h^{s_K}+\sqrt{\frac{\log n}{nh^d}}\right)
\end{align*}
\citep{fan1996,hansen2008}. Here the bounds are understood on a compact interior evaluation region \(\Theta\); obtaining the same orders up to the boundary requires an appropriate boundary correction. The uniform bound also requires the usual kernel regularity and bandwidth conditions, including \(nh^d/\log n\to\infty\).

Balancing the bias and stochastic error gives the respective bandwidth choices
\[
h_2\asymp n^{-1/(2s_K+d)},
\qquad
h_\infty\asymp\left(\frac{\log n}{n}\right)^{1/(2s_K+d)}.
\]
With these choices, respectively,
\begin{align*}
\|\widehat J_n-J\|_{L_2(\Theta)}
&=O_p\!\left(n^{-s_K/(2s_K+d)}\right),\\
\|\widehat J_n-J\|_{L_\infty(\Theta)}
&=O_p\!\left[\left(\frac{\log n}{n}\right)^{s_K/(2s_K+d)}\right].
\end{align*}
In particular, choosing \(q\ge\beta\) exploits the available smoothness of \(J\) and yields
\begin{align*}
\|\widehat J_n-J\|_{L_2(\Theta)}
&=O_p\!\left(n^{-\beta/(2\beta+d)}\right),\\
\|\widehat J_n-J\|_{L_\infty(\Theta)}
&=O_p\!\left[\left(\frac{\log n}{n}\right)^{\beta/(2\beta+d)}\right].
\end{align*}
%
%

\subsection{Local Polynomial Regression}\label{app:lpr}

Local polynomial regression improves upon kernel regression by replacing local averaging with local polynomial approximation \citep{fan1996}. Instead of estimating the expectation by averaging nearby observations, it fits a polynomial locally using kernel-weighted least squares and evaluates the fitted polynomial at the query point. Compared with kernel regression, local polynomial regression enjoys automatic boundary correction and generally achieves smaller approximation bias.

Specifically, let \(p\ge0\) denote the polynomial degree. Around a query point \(\theta\), local polynomial regression solves the weighted least-squares problem
\[
\min_{\bm b}
\sum_{i=1}^n
K_h(\theta-\theta_i)
\left[
F_i-
\sum_{|\alpha|\le p}
b_\alpha(\theta_i-\theta)^\alpha
\right]^2,
\]
where \(\bm b=\{b_\alpha:|\alpha|\le p\}\) denotes the polynomial coefficients. The resulting estimator admits the weighted representation
\[
\widehat J_n(\theta)
=
\sum_{i=1}^n
w_i(\theta)F_i,
\]
where
\[
w(\theta)^T
=
e_1^T
\bigl(R(\theta)^TW(\theta)R(\theta)\bigr)^{-1}
R(\theta)^TW(\theta).
\]
Here \(R(\theta)\) is the local polynomial design matrix, \(W(\theta)\) is the diagonal kernel-weight matrix, and \(e_1=(1,0,\ldots,0)^T\). Thus, local polynomial regression also belongs to the weighted surrogate framework introduced in Section~\ref{sec:ltd}.

The bandwidth determines the size of the local neighborhood, while the kernel function specifies how observations within the neighborhood are weighted. Unlike kernel regression, however, the approximation order is determined by the polynomial degree rather than the kernel order. 

Suppose \(J\in C^\beta(\Theta)\), and let
\(
s_P=\min\{\beta,p+1\}.
\)
Under standard local polynomial regression conditions, the approximation bias is \(O(h^{s_P})\), while the stochastic error is \(O_p((nh^d)^{-1/2})\) under the \(L_2\) norm and \(O_p((\log n/(nh^d))^{1/2})\) under the \(L_\infty\) norm \citep{masry1996uniform, masry1996variance}. Consequently,
\begin{align*}
\|\widehat J_n-J\|_{L_2(\Theta)}
&=
O_p\!\left(
h^{s_P}
+
(nh^d)^{-1/2}
\right),\\
\|\widehat J_n-J\|_{L_\infty(\Theta)}
&=
O_p\!\left(
h^{s_P}
+
\sqrt{\frac{\log n}{nh^d}}
\right).
\end{align*}

Balancing the bias and stochastic error gives the respective bandwidth choices
\[
h_2\asymp n^{-1/(2s_P+d)},
\qquad
h_\infty\asymp\left(\frac{\log n}{n}\right)^{1/(2s_P+d)}.
\]
With these choices, respectively,
\begin{align*}
\|\widehat J_n-J\|_{L_2(\Theta)}
&=
O_p\!\left(
n^{-s_P/(2s_P+d)}
\right),\\
\|\widehat J_n-J\|_{L_\infty(\Theta)}
&=
O_p\!\left[
\left(
\frac{\log n}{n}
\right)^{s_P/(2s_P+d)}
\right].
\end{align*}

In particular, choosing \(p\ge\beta-1\) exploits the available smoothness of \(J\), and yields the optimal convergence rates
\begin{align*}
\|\widehat J_n-J\|_{L_2(\Theta)}
&=
O_p\!\left(
n^{-\beta/(2\beta+d)}
\right),\\
\|\widehat J_n-J\|_{L_\infty(\Theta)}
&=
O_p\!\left[
\left(
\frac{\log n}{n}
\right)^{\beta/(2\beta+d)}
\right].
\end{align*}

%

\subsection{Kernel Ridge Regression}
\label{app:krr}

KRR differs fundamentally from kernel regression and local polynomial regression by constructing a global surrogate rather than a local approximation \citep{scholkopf2002,steinwart2008}. Given a positive definite kernel \(k(\cdot,\cdot)\) with RKHS \(\mathcal H\), KRR learns the surrogate by solving the regularized least-squares problem
\[
\min_{f\in\mathcal H}
\frac1n
\sum_{i=1}^n
\left(
f(\theta_i)-F_i
\right)^2
+
\lambda
\|f\|_{\mathcal H}^2,
\]
where \(\lambda>0\) is the regularization parameter. By the representer theorem \citep{scholkopf2002}, the solution admits the finite-dimensional representation
\[
\widehat J_n(\theta)
=
k(\theta)^T
(K+n\lambda I)^{-1}
F
=
\sum_{i=1}^n
w_i(\theta)F_i,
\]
where \(K=(k(\theta_i,\theta_j))_{i,j=1}^n\) is the kernel matrix,
\(k(\theta)=(k(\theta,\theta_1),\ldots,k(\theta,\theta_n))^T\), and
\[
w(\theta)^T
=
k(\theta)^T
(K+n\lambda I)^{-1}.
\]
Thus, KRR also belongs naturally to the weighted surrogate framework introduced in Section~\ref{sec:ltd}.

The approximation properties of KRR are determined primarily by the choice of reproducing kernel. A particularly important family is the Mat\'ern kernel,
\[
k_\nu(x,y)
=
\frac{2^{1-\nu}}{\Gamma(\nu)}
\left(
\frac{\sqrt{2\nu}\|x-y\|}{\rho}
\right)^\nu
K_\nu
\left(
\frac{\sqrt{2\nu}\|x-y\|}{\rho}
\right),
\]
where \(\nu>0\) is the smoothness parameter, \(\rho>0\) is the length scale, and \(K_\nu\) is the modified Bessel function of the second kind. The corresponding RKHS is norm-equivalent to the Sobolev space \(W^{\nu+d/2,2}\), so \(\nu\) controls the smoothness imposed by the RKHS \citep{wendland2005}. In the limiting case \(\nu\rightarrow\infty\), the Mat\'ern kernel converges to the Gaussian kernel, which is infinitely differentiable \citep{rasmussen2006}.


In addition to the kernel, KRR requires a regularization parameter
\(\lambda\), which controls the balance between approximation and
statistical error. General learning-rate analyses are given by
\citet{caponnetto2007}, while \citet{fischer2020} establish rates in
Sobolev norms, including targets outside the RKHS.

For the setting used in Section~\ref{subsec:ltd-examples}, assume
\(\beta>d/2\) and choose the Mat\'ern smoothness parameter
\(\nu=\beta-d/2\), a fixed length scale, and
\[
\lambda_n\asymp n^{-2\beta/(2\beta+d)}.
\]
The RKHS is then norm-equivalent to \(W^{\beta,2}(\Theta)\),
so \(J\in C^\beta(\Theta)\) belongs to the RKHS.
Applying Theorem 1 of \citet{fischer2020} then yields
\[
\|\widehat J_n-J\|_{L_2(\Theta)}
=O_p\!\left(n^{-\beta/(2\beta+d)}\right).
\]

\subsection{Multiple Kernel Learning}
\label{app:mkl}

Multiple kernel learning (MKL) extends kernel regression in an RKHS by learning a combination of candidate kernels together with the regression fit. For each \(m\in\{1,\ldots,M\}\), the candidate kernel \(k_m\) may depend on a coordinate subset of dimension \(1\le d_m\le d\):
\[
k_m(\theta,\theta')=\bar k_m(\pi_m\theta,\pi_m\theta'),
\qquad \pi_m:\Theta\to\Theta_m\subset\mathbb R^{d_m}.
\]
The dimensions \(d_m\) need not be equal. Given kernels \(k_1,\ldots,k_M\), define
\[
k_\eta(\theta,\theta')
=\sum_{m=1}^{M}\eta_m k_m(\theta,\theta'),
\qquad \eta\in\Delta_M:=\left\{\eta\in[0,\infty)^M:\sum_{m=1}^{M}\eta_m=1\right\}.
\]
The coefficients \(\eta_m\) determine the contribution of each kernel and may be zero. The least-squares MKL estimator solves
\[
(\widehat\eta,\widehat J_n)
\in\arg\min_{\eta\in\Delta_M,\;f\in\mathcal H_{k_\eta}}
\left\{\frac1n\sum_{i=1}^{n}\bigl(F_i-f(\theta_i)\bigr)^2
+\lambda\|f\|_{\mathcal H_{k_\eta}}^2\right\}.
\]
The candidate kernels specify the available representations, while \(\lambda>0\) controls regularization. Learning \(\eta\) allows the representation to adapt to the observed responses. Let \(\mathcal H_m\) be the RKHS of \(k_m\). Minimizing over \(\eta\in\Delta_M\) is equivalent to fitting \(f=\sum_{m=1}^{M}f_m\), \(f_m\in\mathcal H_m\), with penalty \(\lambda(\sum_{m=1}^{M}\|f_m\|_{\mathcal H_m})^2\). Thus this estimator is the squared \(\ell_1\) mixed-norm instance of the framework of \citet{suzuki2011mkl}.

For the learned combination \(\widehat\eta\) held fixed, the fitted surface is a KRR solution:
\[
\widehat J_n(\theta)
=k_{\widehat\eta}(\theta)^T(K_{\widehat\eta}+n\lambda I)^{-1}F
=\sum_{i=1}^{n}w_i^{\rm MKL}(\theta)F_i,
\]
where
\[
K_{\widehat\eta}=\bigl(k_{\widehat\eta}(\theta_i,\theta_j)\bigr)_{i,j=1}^{n},
\qquad
k_{\widehat\eta}(\theta)=\bigl(k_{\widehat\eta}(\theta,\theta_1),\ldots,k_{\widehat\eta}(\theta,\theta_n)\bigr)^T,
\]
and
\[
w^{\rm MKL}(\theta)^T=k_{\widehat\eta}(\theta)^T(K_{\widehat\eta}+n\lambda I)^{-1}.
\]
Because \(\widehat\eta\) depends on the responses, these are response-adaptive weights. The learned kernel combination is held fixed when differentiating with respect to \(\theta\), so differentiability is inherited from the candidate kernels.

For the surface-fitting rate, suppose \(J=\sum_{m=1}^{M}J_m\), with \(J_m\in\mathcal H_m\) and \(\sum_{m=1}^{M}\|J_m\|_{\mathcal H_m}\) bounded. Let \(\Pi\) denote the training-design distribution, set \(d_*:=\max_{1\le m\le M}d_m\), and assume that the component kernel integral operators on \(L_2(\Pi)\) have eigenvalues satisfying \(\mu_{m,j}\le Cj^{-1/\rho_m}\), where \(\rho_m=d_m/(2\beta)\in(0,1)\). We impose Assumptions~1--4 of \citet{suzuki2011mkl}: independent and identically distributed training pairs with correctly specified conditional mean and uniformly bounded noise \(F_i-J(\theta_i)\), bounded kernels and separable RKHSs, the stated spectral decay, and that paper's embedding and incoherence conditions. 
For fixed \(M\), take the auxiliary parameters in \citet[Theorem~1]{suzuki2011mkl} to be \(r_m=n^{-1/[2(1+\rho_m)]}\). With the regularization prescribed by that theorem, which has order \(\lambda_n\asymp n^{-2\beta/(2\beta+d_*)}\), its bound gives
\[
\|\widehat J_n-J\|_{L_2(\Pi)}^2
=O_p\!\left(\sum_{m=1}^{M}n^{-2\beta/(2\beta+d_m)}
+\frac{M\log M}{n}\right).
\]
Since \(M\) is fixed, the largest component dimension determines the leading order. Assuming that the design density is bounded above and away from zero makes the design-weighted and Lebesgue \(L_2\) norms equivalent, giving
\begin{equation}
\label{eq:mkl-component-rate}
\|\widehat J_n-J\|_{L_2(\Theta)}
=O_p\!\left(n^{-\beta/(2\beta+d_*)}\right).
\end{equation}
This is the surface rate transferred to LTD guarantees in Section~\ref{subsec:ltd-examples}. 

\subsection{Smooth Neural Networks}
\label{app:nn-nonparametric}

Smooth neural networks provide a flexible nonlinear approximation to the expectation function. Let \(\sigma\) be a fixed, smooth, nonpolynomial activation, such as sigmoid, and write a network with \(L\) hidden layers and a linear output as
\[
z_0(\theta)=\theta,\qquad
z_\ell(\theta)=\sigma\!\left(W_\ell z_{\ell-1}(\theta)+b_\ell\right),
\quad \ell=1,\ldots,L,
\]
\[
g_{\eta,v}(\theta)=h_\eta(\theta)^Tv,\qquad
h_\eta(\theta)=\bigl(1,z_L(\theta)^T\bigr)^T.
\]
The activation is applied componentwise, \(\eta\) collects the hidden-layer parameters, and \(v\) contains the output coefficients, including the intercept.

Following \citet{park2026}, fix \(L\ge2\) and use a first hidden layer of width at most \((4N+1)^d\), followed by hidden layers of width at most two. Let \(\mathcal N_{N,L}^{\sigma}\) denote this network class. Smaller networks can be embedded in the maximum-width architecture by adding unused neurons, so the fitting procedure can use one prescribed architecture for each \(N\).

For an integer \(\beta>d/2\) and a fixed radius \(R>0\), define
\begin{equation}
\mathcal G_{N,L}(R)=\left\{g\in\mathcal N_{N,L}^{\sigma}:\|g\|_{W^{\beta,\infty}(\Theta)}\le R\right\}.
\label{eq:NN_feasible}	
\end{equation}
The surrogate is fitted by constrained least squares:
\[
\widehat J_n\in\arg\min_{g\in\mathcal G_{N_n,L}(R)}
\frac1n\sum_{i=1}^{n}\bigl(F_i-g(\theta_i)\bigr)^2.
\]
All network weights may be trained. The width controls approximation capacity, while the Sobolev constraint controls the magnitude and smoothness of the fitted function. The radius \(R\) is fixed independently of \(n\), whereas \(N_n\) increases with the sample size. This defines a smooth fitted network without output truncation.

Let \((\widehat\eta,\widehat v)\) represent the fitted network and define the final-feature design matrix
\(
H_{\widehat\eta}=\left(h_{\widehat\eta}(\theta_1)^T,\ldots, h_{\widehat\eta}(\theta_n)^T\right)^\top.
\)
Suppose that this matrix has full column rank and that, with \(\widehat\eta\) fixed, the fitted output vector \(\widehat v\) lies in the interior of its feasible set under the Sobolev constraint. The output-layer first-order condition then gives
\[
\widehat v=(H_{\widehat\eta}^TH_{\widehat\eta})^{-1}H_{\widehat\eta}^TF.
\]
Consequently,
\[
\widehat J_n(\theta)=h_{\widehat\eta}(\theta)^T
(H_{\widehat\eta}^TH_{\widehat\eta})^{-1}H_{\widehat\eta}^TF
=\sum_{i=1}^{n}w_i^{\rm NN}(\theta)F_i,
\]
where
\[
w^{\rm NN}(\theta)^T=h_{\widehat\eta}(\theta)^T
(H_{\widehat\eta}^TH_{\widehat\eta})^{-1}H_{\widehat\eta}^T.
\]
The hidden features are learned from the responses, making these weights response-adaptive. All fitted parameters are held fixed when differentiating with respect to the query point. The rank and output-layer interior conditions justify this representation; they are not required for the surface-rate result below.

For the surface-fitting rate, \cite{park2026} establishes approximation by a constructed network. The following proposition combines that approximation result with least-squares estimation to obtain a statistical rate for the fitted network defined above.

\begin{proposition}
\label{prop:park-nn-surface}
Suppose that \(\{(\theta_i,F_i)\}_{i=1}^{n}\) are independent and identically distributed, with \(|F_i|\le B\) almost surely and
\(
J(\theta)=\mathbb E[F_i\mid\theta_i=\theta].
\)
Let \(\Pi\) denote the marginal distribution of the training input \(\theta_i\). Assume that its density is bounded above and away from zero on \(\Theta\), and that $J\in W^{\beta,\infty}(\Theta)$ and $\|J\|_{W^{\beta,\infty}(\Theta)}\le R_0$
for a fixed integer \(\beta>d/2\). Fix a smooth nonpolynomial activation and a depth \(L\ge2\). Choose \(R\) sufficiently large, as specified in the proof, and let \(N_n\) be deterministic. If \(\widehat J_n\) is a measurable global least-squares minimizer over \(\mathcal G_{N_n,L}(R)\), then
\[
\|\widehat J_n-J\|_{L_2(\Theta)}^2
=O_p\!\left(N_n^{-2\beta}+n^{-2\beta/(2\beta+d)}\right).
\]
In particular, taking \(N_n\asymp n^{1/(2\beta+d)}\) gives
\[
\|\widehat J_n-J\|_{L_2(\Theta)}=O_p\!\left(n^{-\beta/(2\beta+d)}\right).
\]
\end{proposition}

\begin{proof}{Proof}
By \citet[Theorem 2]{park2026}, there exists a network \(g_N\in\mathcal N_{N,L}^{\sigma}\) satisfying simultaneously
\[
\|g_N-J\|_\infty\le C_PR_0N^{-\beta},\qquad
\|g_N-J\|_{W^{\beta,\infty}}\le C_PR_0.
\]
Thus, choosing \(R\ge(1+C_P)R_0\) ensures that \(g_N\in\mathcal G_{N,L}(R)\), and
\[
\inf_{g\in\mathcal G_{N,L}(R)}\|g-J\|_{L_2(\Pi)}^2\le C N^{-2\beta}.
\]

The constrained network class is contained in a fixed-radius \(W^{\beta,\infty}\) ball. The classical H\"older entropy bound \citep[Theorem 2.7.1]{vandervaartwellner1996} therefore gives
\[
\log\mathcal N\bigl(\varepsilon,\mathcal G_{N,L}(R),\|\cdot\|_\infty\bigr)
\le C\varepsilon^{-d/\beta},\qquad 0<\varepsilon\le1,
\]
uniformly in \(N\). The same bound holds in \(L_2(Q)\) for every probability measure \(Q\), since the \(L_2(Q)\) norm is bounded by the uniform norm.

Applying the quadratic-loss oracle inequality of \citet[Section 5.1, Theorem 5.2 and Example 4]{koltchinskii2011} with entropy exponent \(2\rho=d/\beta<2\) yields
\[
\|\widehat J_n-J\|_{L_2(\Pi)}^2
=O_p\!\left(N_n^{-2\beta}+n^{-1/(1+\rho)}\right).
\]
Since \(1/(1+\rho)=2\beta/(2\beta+d)\), equivalence of the design-weighted and Lebesgue norms proves the first assertion. Balancing the two terms gives the stated choice of \(N_n\) and the surface rate.
\hfill\Halmos
\end{proof}

%

\section{Gagliardo--Nirenberg Interpolation Inequalities}
\label{app:gn}

This appendix introduces the function-space notation used
in the general Gagliardo--Nirenberg inequality and derives
four special cases in \(L_p(\Theta)\) norms.
Throughout, \(\Theta\subset\mathbb R^d\) is a compact
hyperrectangle with nonempty interior. 

\subsection{Function Spaces and Norms}

For \(1\le p<\infty\), define
\[
\|f\|_{L_p(\Theta)}
=
\left(\int_\Theta |f(\theta)|^p\,d\theta\right)^{1/p},
\]
and let \(\|f\|_{L_\infty(\Theta)}\) denote the essential
supremum of \(|f|\). For continuous functions on \(\Theta\),
this equals the supremum used in the main text. For
derivatives of order \(r\), write
\[
\|D^{(r)}f\|_{L_p(\Theta)}
=
\max_{|\alpha|=r}
\|D^\alpha f\|_{L_p(\Theta)}.
\]
Here, derivatives are understood in the weak sense on the interior of
\(\Theta\), defined through integration by parts against smooth,
compactly supported test functions. They agree with classical derivatives
for smooth functions.
Correspondingly, for \(1\le p\le\infty\), \(L_p(\Theta)\) denotes the space of
measurable functions \(f:\Theta\to\mathbb R\) with
\(\|f\|_{L_p(\Theta)}<\infty\).

For an integer \(\beta\ge1\) and \(1\le p\le\infty\),
the Sobolev space \(W^{\beta,p}(\Theta)\) consists of functions
whose weak derivatives through order \(\beta\) belong to
\(L_p(\Theta)\), with derivatives taken on the interior of
\(\Theta\). Its norm is
\(\bigl(\sum_{|\alpha|\le\beta}\|D^\alpha f\|_{L_p(\Theta)}^p\bigr)^{1/p}\)
for \(p<\infty\), and
\(\max_{|\alpha|\le\beta}\|D^\alpha f\|_{L_\infty(\Theta)}\)
for \(p=\infty\).

The H\"older space \(C^{k,\gamma}(\Theta)\), where \(k\ge0\)
is an integer and \(0<\gamma\le1\), consists of functions with
continuous derivatives through order \(k\) up to the boundary,
whose derivatives of order \(k\) satisfy
\(|D^\alpha f(x)-D^\alpha f(y)|\le M\|x-y\|^\gamma\)
for all \(x,y\in\Theta\) and some finite \(M\).
For integer smoothness \(\beta\), we use \(C^{\beta-1,1}(\Theta)\):
the derivatives of order \(\beta-1\) are Lipschitz continuous.

The space \(C^\beta(\Theta)\) denotes \(\beta\)-times continuously differentiable functions, with norm
\(
\|f\|_{C^\beta(\Theta)}
=
\max_{|\alpha|\le\beta}
\sup_{\theta\in\Theta}|D^\alpha f(\theta)|.
\)
On the compact hyperrectangle \(\Theta\) considered here,
membership in \(C^\beta(\Theta)\) implies both Sobolev smoothness
\(\beta\), for every \(1\le p\le\infty\), and H\"older smoothness
\(\beta\) in the sense defined above \citep{adams2003}.

\subsection{A General Interpolation Inequality}

The following bounded-domain form of the
Gagliardo--Nirenberg inequality allows the function,
its lower-order derivatives, and its higher-order
derivatives to be measured in different norms
\citep{gagliardo1959,nirenberg1959,adams2003}.
\begin{lemma}[Gagliardo--Nirenberg inequality in $L_p$ norms]
\label{lem:GN-general}
Let \(r\) and \(\beta\) be integers satisfying
\(0\le r<\beta\), and let \(1\le p\le\infty\) and
\(1\le q\le s\le\infty\).
Suppose that \(a\in[r/\beta,1)\) satisfies
\begin{equation}
\frac1p
=
\frac rd
+
a\left(\frac1s-\frac{\beta}{d}\right)
+
\frac{1-a}{q}.
\label{eq:app-GN-balance}
\end{equation}
 Then there exists a constant \(C>0\),
depending only on \(\Theta,d,p,q,s,r,\beta,a\) and not on \(f\),
such that
\begin{equation}
\|D^{(r)}f\|_{L_p(\Theta)}
\le
C\left(
\|D^{(\beta)}f\|_{L_s(\Theta)}^a
\|f\|_{L_q(\Theta)}^{1-a}
+
\|f\|_{L_q(\Theta)}
\right)
\label{eq:app-GN-general}
\end{equation}
for every \(f\in L_q(\Theta)\) with \(D^{(\beta)}f\in L_s(\Theta)\).
\end{lemma}

\subsection{Derivation of the Four Special Cases}

Solving \eqref{eq:app-GN-balance} gives
\[
a=\frac{r+d/q-d/p}{\beta+d/q-d/s}.
\]
For \(0\le r<\beta\), substituting the following four choices
into \eqref{eq:app-GN-general} yields the desired inequalities.
All norms below are over \(\Theta\).

\noindent\textbf{(i)} \(L_2\to L_2\).
Taking \(p=q=s=2\) gives \(a=r/\beta\), hence
\begin{equation}
\|D^{(r)}f\|_{L_2(\Theta)}
\le C\left(
\|D^{(\beta)}f\|_{L_2(\Theta)}^{r/\beta}
\|f\|_{L_2(\Theta)}^{(\beta-r)/\beta}
+\|f\|_{L_2(\Theta)}\right).
\label{eq:app-GN-22}
\end{equation}

\noindent\textbf{(ii)} \(L_\infty\to L_\infty\).
Taking \(p=q=s=\infty\) again gives \(a=r/\beta\), hence
\begin{equation}
\|D^{(r)}f\|_{L_\infty(\Theta)}
\le C\left(
\|D^{(\beta)}f\|_{L_\infty(\Theta)}^{r/\beta}
\|f\|_{L_\infty(\Theta)}^{(\beta-r)/\beta}
+\|f\|_{L_\infty(\Theta)}\right).
\label{eq:app-GN-infinf}
\end{equation}

\noindent\textbf{(iii)} \(L_2\to L_\infty\) using \(\|D^{(\beta)}f\|_{L_2(\Theta)}\).
Taking \(p=\infty\) and \(q=s=2\) gives
\(a=(r+d/2)/\beta\). The condition \(a<1\) requires
\(\beta>r+d/2\), under which
\begin{equation}
\|D^{(r)}f\|_{L_\infty(\Theta)}
\le C\left(
\|D^{(\beta)}f\|_{L_2(\Theta)}^{(r+d/2)/\beta}
\|f\|_{L_2(\Theta)}^{(\beta-r-d/2)/\beta}
+\|f\|_{L_2(\Theta)}\right).
\label{eq:app-GN-2inf-L2}
\end{equation}

\noindent\textbf{(iv)} \(L_2\to L_\infty\) using \(\|D^{(\beta)}f\|_{L_\infty(\Theta)}\).
Taking \(p=s=\infty\) and \(q=2\) gives
\(a=(2r+d)/(2\beta+d)\in[r/\beta,1)\), hence
\begin{equation}
\|D^{(r)}f\|_{L_\infty(\Theta)}
\le C\left(
\|D^{(\beta)}f\|_{L_\infty(\Theta)}^{(2r+d)/(2\beta+d)}
\|f\|_{L_2(\Theta)}^{2(\beta-r)/(2\beta+d)}
+\|f\|_{L_2(\Theta)}\right).
\label{eq:app-GN-2inf-Linf}
\end{equation}

In Lemma~\ref{lem:GN}, \eqref{eq:GN_p} combines cases (i)--(ii)
with \(p=2\) and \(p=\infty\), respectively, while
\eqref{eq:GN_2_inf_L2} and \eqref{eq:GN_2_inf} correspond
to cases (iii) and (iv), respectively.

\section{Derivative Norm Bounds for Representative Surrogate Learning Methods}
\label{app:derivative-bounds}

To apply Theorem~\ref{thm:ltd-general} and Corollary~\ref{cor:l2-linf},
we need to establish \(\|D^{(\beta)}\widehat J_n\|_{L_p(\Theta)}=O_p(1)\),
with \(p\in\{2,\infty\}\) chosen according to the relevant interpolation inequality. 
This section provides proof sketches for these derivative norm bounds for five representative surrogate learning methods: KR, LPR, KRR, MKL and smooth NN.
We assume that $(\theta_i,F_i)_{i=1}^n$ are independent and identically distributed, with
\(
F_i=J(\theta_i)+\varepsilon_i,
\E(\varepsilon_i\mid\theta_i)=0\)
and 
\(\E(\varepsilon_i^2\mid\theta_i)\) is bounded.
The design density is bounded and bounded away from zero on \(\Theta\).
For the uniform bounds for the local methods (KR and LPR), we additionally assume \(\E|F_i|^t<\infty\) for some \(t>2\), together with the conditional moment, kernel regularity, and bandwidth conditions required by the corresponding uniform maximal inequalities, applied to all differentiated kernel averages through order \(\beta\) \citep{hansen2008}.

\subsection{Kernel regression}
\label{app:derivative-kr}

\paragraph{Proof sketch.}
We use the same compact interior evaluation region \(\Theta\) and the same design and kernel conditions as in Appendix~\ref{app:kr}; all norms below are over \(\Theta\).
Write the population smoothing ratio as
\[
R_h(\theta)=
\frac{\int K_h(\theta-x)J(x)p(x)\,dx}
{\int K_h(\theta-x)p(x)\,dx}.
\]
Then
\begin{equation*}
\widehat J_n-J=\underbrace{R_h-J}_{\text{deterministic part}}
+\underbrace{\widehat J_n-R_h}_{\text{stochastic part}}.
\end{equation*}
Locally, integration by parts transfers derivatives from the kernel to \(Jp\) and \(p\). As in Appendix~\ref{app:kr}, interior evaluation and rapid kernel decay make the tail and boundary contributions negligible. Since the derivatives of \(Jp\) and \(p\) through order \(\beta\) are bounded in a fixed neighborhood of \(\Theta\), their local kernel averages are uniformly bounded.
Since the population denominator stays away from zero,
the quotient rule gives
\begin{equation}
\max_{|\alpha|\le\beta}\norm{D^\alpha R_h}_{L_\infty(\Theta)}\le C.
\label{eq:kr-signal}
\end{equation}
For the stochastic part, each derivative costs a factor $h^{-1}$.
Integrated second moments and a uniform maximal inequality with additional tail assumptions, respectively,
give, for $|\alpha|\le\beta$,
\begin{align}
\norm{D^\alpha(\widehat J_n-R_h)}_{L_2(\Theta)}
&=\Op\!\left((nh^{d+2|\alpha|})^{-1/2}\right),
\label{eq:kr-L2-fluctuation}\\
\norm{D^\alpha(\widehat J_n-R_h)}_{L_\infty(\Theta)}
&=\Op\!\left(\sqrt{\frac{\log n}{nh^{d+2|\alpha|}}}\right).
\label{eq:kr-Linf-fluctuation}
\end{align}
Thus the triangle inequality gives
\begin{align*}
\norm{D^{(\beta)}\widehat J_n}_{L_2(\Theta)}
&=\Op\!\left(1+(nh^{d+2\beta})^{-1/2}\right),
\\
\norm{D^{(\beta)}\widehat J_n}_{L_\infty(\Theta)}
&=\Op\!\left(1+\sqrt{\frac{\log n}{nh^{d+2\beta}}}\right).
\end{align*}
The bandwidths
\[
h_2\asymp n^{-1/(2\beta+d)},\qquad
h_\infty\asymp(\log n/n)^{1/(2\beta+d)}
\]
make the respective derivative norms $\Op(1)$. \hfill\Halmos

\subsection{Local polynomial regression}
\label{app:derivative-lpr}

\paragraph{Proof sketch.}
Throughout this proof sketch, all expectations are conditional on the
design points $\theta_1,\ldots,\theta_n$; we suppress the conditioning
in the notation.
Suppose that, with probability tending to one, uniformly over \(\Theta\), the
normalized local design matrices have eigenvalues bounded above and away from zero, and local sample counts are \(O(nh^d)\).
Assume that the kernel and its derivatives through order $\beta$ decay
rapidly, as for the Gaussian kernel. As the bandwidth shrinks, the combined
contribution of observations outside any fixed neighborhood becomes negligible.
On the stated event, differentiating the local normal equations contributes
a factor $h^{-1}$ per derivative. Together with the local sample-count
bound and the rapid kernel decay, this gives
\begin{equation}
\sum_i|D^\alpha w_i(\theta)|\le Ch^{-|\alpha|},
\qquad
\sum_i|D^\alpha w_i(\theta)|^2
\le\frac{C}{nh^{d+2|\alpha|}},
\quad |\alpha|\le\beta.
\label{eq:lpr-weight-sums}
\end{equation}
To control the conditional mean, fix $\theta_0$ and hold the Taylor
polynomial of $J$ at $\theta_0$ through degree $\beta-1$ fixed while
the query point varies. Since $p\ge\beta-1$, the weights reproduce
this polynomial exactly. Differentiating this identity and evaluating
at $\theta_0$ gives
 \begin{equation*}
 \left.D_\theta^\alpha
 \E[\widehat J_n(\theta)]
 \right|_{\theta=\theta_0}
 =D^\alpha J(\theta_0)\mathbf 1\{|\alpha|<\beta\}
 +\sum_iD^\alpha w_i(\theta_0)
 \left[J(\theta_i)-\sum_{|\nu|\le\beta-1}
 \frac{D^\nu J(\theta_0)}{\nu!}(\theta_i-\theta_0)^\nu\right].
 \end{equation*}
The bracketed remainder is $O(\|\theta_i-\theta_0\|^\beta)$.
Although it grows with distance, the rapid decay of the differentiated
weights keeps its total contribution bounded by $Ch^{\beta-|\alpha|}$.
Hence
\begin{equation}
\max_{|\alpha|\le\beta}
\left\|D^\alpha\E[\widehat J_n]\right\|_{L_\infty(\Theta)}
\le C.
\label{eq:lpr-signal}
\end{equation}
Finite volume gives the corresponding derivative $L_2(\Theta)$ bound. 

For the centered noise, bounded conditional variance and
\eqref{eq:lpr-weight-sums} give
\begin{equation*}
\left\|D^\alpha\left\{\widehat J_n
-\E[\widehat J_n]\right\}\right\|_{L_2(\Theta)}
=\Op\!\left((nh^{d+2|\alpha|})^{-1/2}\right).
\end{equation*}
With additional tail and kernel
regularity conditions, the conditional maximal inequality gives
\begin{equation}
\left\|D^{(\beta)}\left\{\widehat J_n
-\E[\widehat J_n]\right\}\right\|_{L_\infty(\Theta)}
=\Op\!\left(\sqrt{\frac{\log n}{nh^{d+2\beta}}}\right).
\label{eq:lpr-noise-Linf}
\end{equation}
Combining the conditional mean and noise bounds yields
\begin{align*}
\norm{D^{(\beta)}\widehat J_n}_{L_2(\Theta)}
&=\Op\!\left(1+(nh^{d+2\beta})^{-1/2}\right),
\\
\norm{D^{(\beta)}\widehat J_n}_{L_\infty(\Theta)}
&=\Op\!\left(1+\sqrt{\frac{\log n}{nh^{d+2\beta}}}\right).
\end{align*}
At $h_2\asymp n^{-1/(2\beta+d)}$ and
$h_\infty\asymp(\log n/n)^{1/(2\beta+d)}$, the respective derivative norms are $\Op(1)$. \hfill\Halmos

\subsection{Kernel ridge regression}
\label{app:derivative-krr}
For the fixed Mat\'ern kernel with $\nu=\beta-d/2>0$ specified in
Appendix~\ref{app:krr}, the RKHS--Sobolev norm equivalence gives
\begin{equation}
\norm{D^\alpha g}_{L_2(\Theta)}
\le\norm{g}_{W^{\beta,2}(\Theta)}\le C_D\norm{g}_{\Hcal},
\qquad |\alpha|\le\beta,\quad g\in\Hcal,
\label{eq:krr-derivative-operator}
\end{equation}
where $C_D$ depends only on the fixed kernel and domain.
Correspondingly, we also use the eigenvalue bound $\mu_j\lesssim j^{-2\beta/d}$,
which holds for the population covariance operator of the Mat\'ern kernel under the stated design conditions.
\paragraph{Proof sketch.}
The empirical covariance operator is
\[
T_n=\frac1n\sum_i k(\theta_i,\cdot)\otimes k(\theta_i,\cdot).
\]
The normal equation is
$(T_n+\lambda I)\widehat J_n=n^{-1}\sum_iF_i k(\theta_i,\cdot)$.
Substituting $F_i=J(\theta_i)+\varepsilon_i$ and subtracting $J$ gives
\begin{equation}
\widehat J_n-J=-\lambda(T_n+\lambda I)^{-1}J
+(T_n+\lambda I)^{-1}\frac1n\sum_i\varepsilon_i k(\theta_i,\cdot).
\label{eq:krr-decomp}
\end{equation}
The first term is the conditional bias. Since $T_n$ is positive and
self-adjoint, $\lambda(T_n+\lambda I)^{-1}$ has operator norm at most
one: its spectral multiplier is $\lambda/(t+\lambda)\le1$.
Thus the derivative operator bound gives
\begin{equation*}
\left\|D^{(\beta)}\{\lambda(T_n+\lambda I)^{-1}J\}\right\|_{L_2(\Theta)}
\le C_D\left\|\lambda(T_n+\lambda I)^{-1}J\right\|_{\Hcal}
\le C_D\norm{J}_{\Hcal}.
\end{equation*}

For the stochastic term, conditional on the design,
\begin{equation}
\E\!\left[
\left\|(T_n+\lambda I)^{-1}\frac1n\sum_i\varepsilon_i k(\theta_i,\cdot)
\right\|_{\Hcal}^2\,\right]\le\frac Cn\operatorname{tr}[T_n(T_n+\lambda I)^{-2}]\le\frac C{n\lambda}
\operatorname{tr}[T_n(T_n+\lambda I)^{-1}].
\label{eq:krr-noise}
\end{equation}
The last inequality uses $t/(t+\lambda)^2\le
\lambda^{-1}t/(t+\lambda)$. The map
$A\mapsto A(A+\lambda I)^{-1}$ is operator concave, so Jensen's
inequality and the population eigenvalue bound imply
\[
\E\operatorname{tr}[T_n(T_n+\lambda I)^{-1}]
\le\sum_{j\ge1}\frac{\mu_j}{\mu_j+\lambda}
\lesssim\lambda^{-d/(2\beta)},\qquad 0<\lambda\le1.
\]
Hence, \eqref{eq:krr-derivative-operator} gives
\[
\left\|D^{(\beta)}\left\{
(T_n+\lambda I)^{-1}\frac1n\sum_i\varepsilon_i k(\theta_i,\cdot)
\right\}\right\|_{L_2(\Theta)}
=\Op\!\left(\{n\lambda^{1+d/(2\beta)}\}^{-1/2}\right).
\]
Therefore
\begin{equation*}
\norm{D^{(\beta)}\widehat J_n}_{L_2(\Theta)}
=\Op\!\left(1+\{n\lambda^{1+d/(2\beta)}\}^{-1/2}\right).
\end{equation*}
At $\lambda_n\asymp n^{-2\beta/(2\beta+d)}$,
$n\lambda_n^{1+d/(2\beta)}\asymp1$, so the derivative norm is $\Op(1)$. \hfill\Halmos

\subsection{Multiple kernel learning}
\label{app:derivative-mkl}
For fixed $M$, consider the estimator
\begin{equation}
\min_{f_m\in\Hcal_m}
\left\{\frac1n\sum_i\left(F_i-\sum_m f_m(\theta_i)\right)^2
+\lambda\left(\sum_m\norm{f_m}_{\Hcal_m}\right)^2\right\}.
\label{eq:mkl-objective}
\end{equation}
Let $(\widehat f_m)_{m=1}^M$ be a minimizer and set $\widehat J_n=\sum_{m=1}^M\widehat f_m$.
We adopt the kernel and design conditions of Appendix~\ref{app:mkl},
with $d_*:=\max_m d_m$ and $\beta>d_*/2$.
Assume a fixed decomposition $J=\sum_mJ_m$, $J_m\in\mathcal H_m$.
In addition, assume that functions in each $\mathcal H_m$ have weak
derivatives through order $\beta$ satisfying
\begin{equation}
\|D^\alpha \widehat f_m\|_{L_2(\Theta_m)}\le C_m\|\widehat f_m\|_{\mathcal H_m},
\qquad |\alpha|\le\beta,\quad \widehat f_m\in\mathcal H_m.
\label{eq:mkl-derivative-operator}
\end{equation}

\paragraph{Proof sketch.}
Comparing the optimizer in \eqref{eq:mkl-objective} with the fixed
true decomposition, and substituting $F_i=J(\theta_i)+\varepsilon_i$,
gives the basic inequality
\begin{equation}
\frac1n\sum_i\{\widehat J_n(\theta_i)-J(\theta_i)\}^2
+\lambda\left(\sum_m\norm{\widehat f_m}_{\Hcal_m}\right)^2\le\frac2n\sum_i\varepsilon_i
\{\widehat J_n(\theta_i)-J(\theta_i)\}
+\lambda\left(\sum_m\norm{J_m}_{\Hcal_m}\right)^2.
\label{eq:mkl-basic}
\end{equation}

Use the RKHS $\Hcal_+$ of the fixed sum kernel $k_+=\sum_mk_m$, with
empirical covariance operator
\[
T_{n,+}=\frac1n\sum_i k_+(\theta_i,\cdot)\otimes k_+(\theta_i,\cdot).
\]
The sum-kernel norm is bounded by the sum of the component norms, so
\[
\norm{\widehat J_n-J}_{\Hcal_+}^2
\le2\left(\sum_m\norm{\widehat f_m}_{\Hcal_m}\right)^2
+2\left(\sum_m\norm{J_m}_{\Hcal_m}\right)^2.
\]
Cauchy--Schwarz with the positive operator $T_{n,+}+\lambda I$ bounds the noise cross term by
\begin{equation*}
\begin{aligned}
\frac2n\sum_i\varepsilon_i
\{\widehat J_n(\theta_i)-J(\theta_i)\}
&\le\frac14\left\{
\frac1n\sum_i\{\widehat J_n(\theta_i)-J(\theta_i)\}^2
+\lambda\norm{\widehat J_n-J}_{\Hcal_+}^2\right\}\\
&\qquad+4\left\langle
\frac1n\sum_i\varepsilon_i k_+(\theta_i,\cdot),
(T_{n,+}+\lambda I)^{-1}
\frac1n\sum_i\varepsilon_i k_+(\theta_i,\cdot)
\right\rangle_{\Hcal_+}.
\end{aligned}
\end{equation*}
Substituting this into \eqref{eq:mkl-basic} 
yields
\begin{equation*}
\left(\sum_m\norm{\widehat f_m}_{\Hcal_m}\right)^2
\le3\left(\sum_m\norm{J_m}_{\Hcal_m}\right)^2 +\frac8\lambda\left\langle
\frac1n\sum_i\varepsilon_i k_+(\theta_i,\cdot),
(T_{n,+}+\lambda I)^{-1}
\frac1n\sum_i\varepsilon_i k_+(\theta_i,\cdot)
\right\rangle_{\Hcal_+}.
\end{equation*}

Conditional on the design, independence and bounded error variance
bound the expectation of the quadratic form by
\[
\frac Cn\operatorname{tr}[T_{n,+}(T_{n,+}+\lambda I)^{-1}].
\]

The effective dimension of the sum kernel is bounded by the sum of
the component effective dimensions. Combining this bound with the
operator concavity argument used for KRR gives, for $0<\lambda\le1$,
\[
\E\operatorname{tr}[T_{n,+}(T_{n,+}+\lambda I)^{-1}]
\le\sum_m\sum_{j\ge1}\frac{\mu_{m,j}}{\mu_{m,j}+\lambda}
\lesssim\sum_m\lambda^{-d_m/(2\beta)}
\lesssim\lambda^{-d_*/(2\beta)}.
\]
The last step uses fixed $M$ and $d_m\le d_*$.
Therefore,
\begin{equation*}
\left(\sum_m\norm{\widehat f_m}_{\Hcal_m}\right)^2
=\Op\!\left(1+\frac1{n\lambda^{1+d_*/(2\beta)}}\right).
\end{equation*}

The component bounds in \eqref{eq:mkl-derivative-operator}
extend to $\Theta$ up to fixed volume factors.
The triangle inequality then gives
\[
\norm{D^{(\beta)}\widehat J_n}_{L_2(\Theta)}
\le C\sum_m\norm{\widehat f_m}_{\Hcal_m},\qquad
\norm{D^{(\beta)}J}_{L_2(\Theta)}
\le C\sum_m\norm{J_m}_{\Hcal_m}.
\]
Consequently,
\begin{equation*}
\begin{aligned}
\norm{D^{(\beta)}\widehat J_n}_{L_2(\Theta)}
+\norm{D^{(\beta)}(\widehat J_n-J)}_{L_2(\Theta)}
\le C\left(\sum_m\norm{\widehat f_m}_{\Hcal_m}
+\sum_m\norm{J_m}_{\Hcal_m}\right)=\Op\!\left(1+\{n\lambda^{1+d_*/(2\beta)}\}^{-1/2}\right).
\end{aligned}
\end{equation*}
At $\lambda_n\asymp n^{-2\beta/(2\beta+d_*)}$,
$n\lambda_n^{1+d_*/(2\beta)}\asymp1$, so the derivative norm is
$\Op(1)$. \hfill \Halmos

\subsection{Smooth neural networks}
\label{app:derivative-nn}
\paragraph{Proof sketch.}
Feasibility \eqref{eq:NN_feasible} gives, for every sample,
\begin{equation*}
\norm{D^{(\beta)}\widehat J_n}_{L_\infty(\Theta)}\le R.
\end{equation*}
Finite volume also gives
$\norm{D^{(\beta)}\widehat J_n}_{L_2(\Theta)}\le|\Theta|^{1/2}R$.
Thus both required derivative norms are  $O(1)$. \hfill \Halmos

\section{Geometric Asian Option Example}
\label{app:asian}

This appendix supplements Section~\ref{subsec:numerical-asian} with the benchmark estimators and implementation choices.

\subsection{Benchmark methods}

For a single option, geometric Brownian motion gives
$S_t=x\exp\{(r-\sigma^2/2)t+\sigma W_t\}$. Thus the entire path,
and hence its geometric average, scales with the initial price:
$G_m(x)=xA_m$, where
$A_m=G_m(x)/x$ is a random multiplier that does not depend on
$x$ when the underlying Brownian path is held fixed. For each simulated path, we compute $A_m$ by dividing its
monitored geometric average by the initial price.
The logarithm of this random multiplier is normal, with
\[
\mu_m=\mathbb E[\log A_m]=(r-\sigma^2/2)\frac{T(m+1)}{2m},
\qquad
v_m=\operatorname{Var}(\log A_m)
=\sigma^2\frac{T(m+1)(2m+1)}{6m^2}.
\]
Consequently, the price and its first two derivatives are available in
closed form \citep{glasserman2004}. With
$z=(\log(x/K)+\mu_m+v_m)/\sqrt{v_m}$ and
$c_m=\exp(-rT+\mu_m+v_m/2)$, they are
\begin{align*}
V_m(x)&=xc_m\Phi(z)-Ke^{-rT}\Phi(z-\sqrt{v_m}),\\
V'_m(x)&=c_m\Phi(z),\qquad
V''_m(x)=\frac{c_m\phi(z)}{x\sqrt{v_m}}.
\end{align*}
Here $\Phi$ and $\phi$ are the standard normal distribution and density
functions. The portfolio price is $V_{d,m}(\mathbf x)=\sum_{j=1}^dV_m(x_j)$,
so its $j$th Delta and diagonal Gamma are $V'_m(x_j)$ and $V''_m(x_j)$;
mixed second derivatives vanish. These formulas supply the evaluation
benchmarks only. Training and comparator observations are generated by
simulating the asset paths at all $m$ monitoring dates.

The pathwise method differentiates each simulated payoff before averaging
\citep{broadie1996,glasserman2004}. Since
$\partial G_m(x)/\partial x=A_m$, its Delta estimator is
\[
\widehat\Delta_{\rm PW}(x)
=\frac{e^{-rT}}{N}\sum_{i=1}^N A_{m,i}\mathbf1\{xA_{m,i}>K\}.
\]
A further differentiation encounters the indicator's discontinuity at
the strike. For Gamma, we therefore smooth this boundary contribution
with a Gaussian kernel, following \citet{liu2011}:
\[
\widehat\Gamma_{\rm PW}(x)
=\frac{e^{-rT}}{Nx h_A}\sum_{i=1}^N
A_{m,i}^2\phi\!\left(\frac{A_{m,i}-K/x}{h_A}\right),
\qquad h_A=(4/3)^{1/5}\widehat\sigma_A N^{-1/5}.
\]
The bandwidth $h_A$ is measured on the scale of $A_m$.
The sample standard deviation $\widehat\sigma_A$ of $A_m$ is computed from the
same $N=50{,}000$ paths used for estimation, without additional pilot paths.

For the geometric Asian option, an LR estimator could be derived
directly from the lognormal density of $G_m(x)$. This construction,
however, relies on a special property of the geometric average and
does not extend directly to the more commonly used arithmetic Asian option, whose average
lacks a closed-form density. We choose the geometric Asian option
because its exact price and derivatives provide evaluation benchmarks.
For LR estimation, we use the more general state-space enlargement
approach, taking the entire monitored asset path as the state.

By the Markov property of geometric Brownian motion, the joint
path density factors into transition densities. Only the first
transition depends explicitly on the initial price $x$ and
therefore determines the LR weights.
Let $\Delta t=T/m$ and
define the standardized first Brownian increment
\[
Z_1=
\frac{\log(S_{\Delta t}/x)-(r-\sigma^2/2)\Delta t}
{\sigma\sqrt{\Delta t}},
\qquad Z_1\sim N(0,1).
\]
The LR weights for Delta and Gamma, respectively, are
\[
\frac{Z_1}{x\sigma\sqrt{\Delta t}},
\qquad
\frac{Z_1^2-1}{x^2\sigma^2\Delta t}
-\frac{Z_1}{x^2\sigma\sqrt{\Delta t}}.
\]
Each LR estimate averages the corresponding weight times the
discounted payoff $e^{-rT}(xA_m-K)^+$.
These same path-density weights can be applied to arithmetic Asian
payoffs without knowing the distribution of their averages.
The factors involving $\Delta t$ help explain why the LR estimators
can become more variable as monitoring becomes more frequent.

For the portfolio, each benchmark uses the payoff of the asset whose
sensitivity is being estimated. Within a replication, paths are shared
across test points and between PW and LR. Pathwise Delta and LR require
no bandwidth selection.

\subsection{LTD estimator}\label{app:local-protocol}

Each method fits the portfolio price as a function of the full vector
of initial prices. KR and LPR construct a local fit near each query,
whereas KRR, MKL, and NN fit one global model to the training sample.
KRR, MKL, and NN share the same sites and simulated responses within
each replication. Although the portfolio price is additive, none of
the five learners is given an additive representation or a shared
one-asset pricing function. This keeps the structural information available to the
learners comparable, so that performance differences reflect their
ability to learn from the simulated portfolio responses directly.

In both examples, training sites are fixed across replications
for each setting and design family, while simulation responses
are generated independently in each replication.
All tuning uses the current replication's simulated responses,
without additional pilot simulations or derivative observations.
GCV balances response fit against model complexity
\citep{golub1979generalized}.
After fitting, we hold the selected parameters fixed and
differentiate the fitted surface. When inputs or responses have
been rescaled, we apply the chain rule to obtain derivatives
of the portfolio price with respect to the original initial prices.

\subsubsection{KR and LPR.}

KR and LPR fit the price locally around each query and then
differentiate the fitted surface. The main choices are how to
weight the training responses and how large a neighborhood to use.
We select bandwidths by GCV and enlarge them when needed to avoid
unstable derivative estimates. LPR also uses a small ridge adjustment
to stabilize the local polynomial fit. These steps are detailed below.

\paragraph{Local fits and differentiation.}
KR forms a normalized, locally weighted combination of responses,
using a tensor-product fourth-order kernel with factor
$K_4(t)=\tfrac12(3-t^2)\phi(t)$. This higher-order kernel can assign
negative weights to reduce smoothing bias. LPR fits a Gaussian-weighted quadratic polynomial, including cross terms, near each query. Derivatives can be estimated either from
the polynomial coefficients or by differentiating the fitted
response surface; these two estimates need not coincide
\citep{fan1996,racine2016}. We use the latter approach to implement
LTD. Specifically, the fitted response at each query is the local
polynomial's intercept. We differentiate this intercept as a
function of the query, accounting for changes in both the local
weights and the fitted coefficients, while holding the selected
bandwidths fixed.

\paragraph{Bandwidth selection.}
KR and LPR use the rescaled coordinates $z_j=(x_j-100)/30$,
with 30 as a fixed reference scale. The bandwidths and
bandwidth-enlargement conditions are expressed
in these coordinates.

Bandwidths determine the size of the neighborhood used by KR and LPR.
In each replication, we randomly select up to 200 training sites
and their simulated sample means for bandwidth selection, using
all sites if fewer than 200 are available. For each candidate
bandwidth vector, we fit the local estimator on this subset and
compute its response-based GCV score. The same subset is used
throughout the search, and the vector with the lowest score is
selected. Restricting this step
to a subset reduces the cost of repeated local fits.

The initial bandwidth grid is
$\{0.04,0.06,0.09,0.135,0.2,0.3,0.45,0.675,1,1.5,2.25,3.4\}$.
We first search a common bandwidth, then adjust each coordinate once
using factors $\{0.5,0.75,1,4/3,2\}$.
Candidate bandwidths in this coordinate search are restricted
to $[0.04,3.4]$; the subsequent safeguards may enlarge them
beyond this interval.

\paragraph{Bandwidth safeguards for derivatives.}
Because GCV only assesses response prediction, its selected bandwidths
for KR and LPR may be too small for stable derivative estimation.
We therefore check the selected bandwidths against coverage and
derivative-stability conditions using the full number of training
sites and the total simulation budget.

Let $\widetilde h=(\widetilde h_1,\ldots,\widetilde h_d)$ denote
the bandwidth vector selected by GCV. We use the final bandwidth
vector $h=c\widetilde h$, where $c\ge1$ is the smallest common
multiplier for which
\[
n\prod_{k=1}^d h_k\ge(\log n)^2,\qquad
B\Bigl(\prod_{k=1}^d h_k\Bigr)h_j^{2r_*}\ge(\log B)^2
\quad(j=1,\ldots,d).
\]
Here $n$ is the full number of training sites, $B$ is the total
simulation budget, and $r_*=2$ is the highest derivative order
estimated in the option experiment. If the GCV-selected bandwidths
already satisfy both conditions, we take $c=1$.
The first condition guards against neighborhoods containing too
few training sites: for a regular design, the local site count
scales with $n\prod_k h_k$. The second reflects the variance
scaling of an $r_*$th derivative estimate, which involves the
factor $\{B(\prod_k h_k)h_j^{2r_*}\}^{-1}$.
It therefore prevents bandwidths from becoming too small relative
to the simulation budget and the highest derivative order $r_*$.
The final predictor and its derivatives use all training sites
with the resulting bandwidths.

\paragraph{Numerical stability of LPR.}
LPR also uses a small ridge adjustment to stabilize the local
least-squares solve. At each query, let $A=X^\top W X$ be the
local normal matrix, where $X$ contains the polynomial terms
and $W$ contains their observation weights. We replace $A$ by
\[
A+\varepsilon\,\operatorname{diag}(0,1,\ldots,1).
\]
The first entry corresponds to the intercept, which is left
unpenalized. The remaining entries correspond to the linear
and quadratic terms. We set $\varepsilon$ to $10^{-8}$ times
the mean of their diagonal entries in $A$, so the adjustment
scales with the local fitting problem.

\subsubsection{KRR and MKL.}\label{app:b-kernel}

KRR fits the whole price surface with one Gaussian kernel.
MKL learns a weighted combination of three Gaussian kernels with
different length scales. For MKL, we place a lower bound on each
weight and rescale the kernels to account for their different
derivative scales. We explain these choices and show how the
combined kernel is used for fitting and differentiation. We then
describe how length scales and regularization parameters are
selected from the simulated responses.

\paragraph{Kernels and response centering.}
KRR uses a Gaussian kernel depending jointly on all input coordinates.
Each coordinate has its own length scale: shorter scales allow more
rapid variation, whereas longer scales produce smoother fits.
MKL combines three such kernels, with length vectors
$\ell/2,\ell,2\ell$, and learns their relative contributions by
penalized least squares.

KRR and MKL use $u_j=(x_j-100)/50$, which maps the training
box $[50,150]^d$ to $[-1,1]^d$. Their kernel length scales
are expressed in these coordinates. The full Gaussian kernel is
\[
k_\ell(u,u')=\exp\!\left\{-\frac12\sum_{j=1}^d
\frac{(u_j-u'_j)^2}{\ell_j^2}\right\}.
\]

Let $y_i$ be the simulated sample mean at
training site $i$, and let $\bar y=n^{-1}\sum_{i=1}^n y_i$.
We fit the centered responses $y_i-\bar y$ and add $\bar y$
back to each prediction. This separates the overall price level
from the variation fitted by the kernels. The added constant
does not contribute to the estimated derivatives.

\paragraph{Regularizing the mixture weights.}
For MKL, the three kernel weights sum to one and are each at least $1/6$.
Equivalently, half of the total weight is allocated equally across
the three scales, while the remaining half is learned from the
responses. This restriction prevents the fit from concentrating
entirely on one scale. It is intended to stabilize derivative
estimation, since a kernel combination that fits noisy responses
well can still produce small-scale fluctuations that are amplified
by differentiation. The lower bound of $1/6$ on each kernel weight is a fixed
regularization choice used in both the option and wireless experiments.

\paragraph{Normalizing the candidate kernels.}
Before learning the mixture weights, MKL rescales each candidate
kernel by a positive constant computed from its function and
derivative scales. This rescales the kernel values and their
derivatives together. The purpose is to regularize kernels more
strongly when their derivatives are large relative to their
function scale. This kernel normalization is separate from
scaling the inputs and responses. For $a=1,2,3$, let $k_a$ have length vector
$\ell_a\in\{\ell/2,\ell,2\ell\}$, and define its Gram matrix by
$(K_a)_{ij}=k_a(u_i,u_j)$.
Set $P=I_n-\mathbf1_n\mathbf1_n^\top/n$, where $I_n$ is the
$n\times n$ identity matrix and $\mathbf1_n$ is the vector
of $n$ ones. The centered kernel matrix is $PK_aP$, and $\tau_a=\operatorname{tr}(PK_aP)/n$ is its average
diagonal entry. Dividing by $\tau_a$ alone would put the three
centered kernel matrices on the same function scale, with an
average diagonal entry of one. To also account for derivatives, define
\[
D_a=
\sum_{1\le|\nu|\le3}
\left.
\partial_u^\nu\partial_{u'}^\nu k_a(u,u')
\right|_{u'=u}.
\]
Here $\nu=(\nu_1,\ldots,\nu_d)$ is a vector of nonnegative
integers, where $\nu_j$ specifies the number of derivatives
taken in coordinate $j$, and $|\nu|=\sum_j\nu_j$.
For Gaussian kernels, each term is nonnegative and depends
only on the length vector $\ell_a$.
For example, in one dimension,
$D_a=\ell_a^{-2}+3\ell_a^{-4}+15\ell_a^{-6}$,
so shorter length scales give larger values of $D_a$. We use the rescaled matrices
\[
\widetilde K_a
=c\,\frac{PK_aP}{\tau_a+D_a},
\qquad
c=\left(
\frac13\sum_{b=1}^3\frac{\tau_b}{\tau_b+D_b}
\right)^{-1}.
\]
Relative to normalization by $\tau_a$ alone, the additional
factor $\tau_a/(\tau_a+D_a)$ reduces the contribution of kernels
with larger derivative scales. The common factor $c$ sets the
average scale of the three matrices to one:
$\frac13\sum_{a=1}^3\operatorname{tr}(\widetilde K_a)/n=1$.
All these factors are computed from the training inputs and
kernel length scales, without using derivative observations. 

\paragraph{Fitting and differentiation.}
The matrices $\widetilde K_a$ then replace the original candidate
matrices in the MKL fit. Specifically, MKL uses the combined matrix
\[
K_\eta=\sum_{a=1}^3\eta_a\widetilde K_a,
\qquad
\sum_{a=1}^3\eta_a=1,\quad \eta_a\ge\frac16.
\]
The factors $c/(\tau_a+D_a)$ are fixed before learning the mixture
weights $\eta_a$ from the simulated responses by penalized least
squares. Thus, kernel normalization and mixture-weight estimation
are separate steps.

For fixed mixture weights and regularization parameter $\lambda$,
the regression coefficients are
\[
\alpha=(K_\eta+n\lambda I_n)^{-1}
       (y-\bar y\mathbf1_n),
\]
where $y$ is the vector of simulated sample means at the training
sites and $\bar y$ is their average. At a new input $u$, let
$\widetilde k_a(u)$ be the vector of kernel evaluations against
the training sites, centered using the training-sample averages
and multiplied by the same factor $c/(\tau_a+D_a)$.
The fitted response is
\[
\widehat f(u)=\bar y+
\left(\sum_{a=1}^3\eta_a\widetilde k_a(u)\right)^\top\alpha.
\]
We obtain derivative estimates by differentiating this fitted
function with respect to $u$, holding all fitted parameters and
normalization factors fixed, and then converting the derivatives
to the original input units. The normalization therefore enters
both the regression fit and the kernel derivatives used for
prediction; it is not an adjustment applied after estimating
the derivatives.

In both the option and wireless experiments, MKL uses the
normalized candidate matrices $\widetilde K_a$ defined above,
with $D_a$ including derivative orders one through three.
KRR uses the original Gaussian kernel without this normalization.

\paragraph{Length-scale and regularization selection.}
For the option experiment, KRR and MKL use robust GCV
\citep{lukas2006robust} to discourage fits that respond too
strongly to simulation noise, particularly when estimating Gamma.
The wireless experiment follows the same search procedure using
ordinary GCV. Parameter selection uses only simulated responses;
the exact Deltas and Gammas play no part in parameter selection.

For KRR and MKL, we first select the base length vector $\ell$
using a single-Gaussian KRR fit. We start with a common length
across all coordinates, chosen from
$\{0.35,0.5,0.7,1,1.4,2,2.8,4,5.6,8,12\}$.
For each candidate length, we select the matrix shift
$\delta=n\lambda$ from 13 logarithmically spaced values between
$10^{-5}$ and $10$, and retain the length and shift with the
lowest robust GCV score.
Starting from this common length, we adjust each coordinate once
by comparing multiples $\{0.5,0.8,1,1.25,2\}$ of its current
length, keeping the other coordinates fixed and restricting
lengths to $[0.2,12]$. For each candidate vector, we search the
same shift grid and retain the choice with the lowest score.

The resulting vector $\ell$ defines the KRR kernel and the three
MKL kernels with length vectors $\ell/2,\ell,2\ell$.
With these kernels fixed, KRR and MKL select their regularization
parameters separately. KRR uses the matrix-shift grid described
above. MKL divides each value in this grid by $\tau_2$, the
function scale of its middle kernel at length vector $\ell$,
and uses the resulting values as candidates for $n\lambda$.
All searches use the full training sample.

The robust GCV criterion used for the option experiment is
\[
\operatorname{RGCV}
=
\left\{0.1+0.9\frac{\operatorname{tr}(H^\top H)}{n}\right\}
\frac{n^{-1}\|y-\widehat y\|^2}
{\{1-\operatorname{tr}(H)/n\}^2},
\]
where $y$ collects the site means and $\widehat y$ contains their
fitted values. The last fraction is ordinary GCV; the additional
factor penalizes sensitivity of the fitted values to the responses.
For KRR, $H$ is the response smoother matrix.
For MKL, $H$ is the local Jacobian of the fitted values with
respect to $y$, accounting for changes in the learned mixture
weights subject to their active constraints.
Both matrices include the contribution of the fitted mean.
Kernel systems are solved by dense factorization.

\subsubsection{Smooth NN.}\label{app:b-nn}

We fit a smooth neural network to the simulated responses and
differentiate the fitted network. To limit oscillations in the
fitted surface, training bounds the network's function values and
derivatives at numerical constraint points after standardizing the
inputs and responses. We specify the network architecture and these
constraints, then explain how the network is initialized and trained.
GCV selects the initialization parameters, and errors on held-out
responses determine how many training rounds to use.

\paragraph{Network architecture.}
The network architecture follows Appendix~\ref{app:nn-nonparametric}.
Smooth activation functions allow the fitted surface to be
differentiated.

The network has layer widths $[d,(4N_{\mathrm{NN}}+1)^d,2,1]$,
corresponding to the input layer, two hidden layers, and a
scalar output. We set
\[
s=\max\{3,\lfloor d/2\rfloor+1\},\qquad
N_{\mathrm{NN}}
=\max\{1,\lfloor n^{1/(2s+d)}\rfloor\},
\]
where $n$ is the number of training sites before the validation
split. Thus, $s=3$ for all dimensions in the option experiment
and for $d=2,4$ in the wireless experiment; $s=4$ for the
six-dimensional wireless experiment.
The parameter $N_{\mathrm{NN}}$ determines the first hidden-layer
width and is computed from $n$ and $d$, without tuning.
The same width is used during validation and refitting on all sites. The first hidden layer uses $\tanh(t)$, and the second
uses $20\tanh(t/20)$. The latter is smooth and nearly linear
near zero, which facilitates initialization from the
random-feature regression. The output is a linear combination of the second-layer outputs
plus an intercept, with no clipping.

\paragraph{Validation split and standardization.}
In each replication, we randomly reserve 20\% of the training sites
for validation and use the remaining 80\% to fit the network.
Both inputs and responses are standardized before training.
Each input coordinate is centered at the midpoint of its training
range and divided by the half-width of that range, mapping the
training box to $[-1,1]^d$. The simulated sample means at the
training sites are centered by their average and divided by
their standard deviation. These response statistics are computed from the fitting subset
during validation and from all sites for the final fit.

\paragraph{Constrained fitting and numerical checks.}
Let $\mathcal F_{\mathrm{NN}}$ denote the network class with the
architecture specified above. Training seeks to minimize the
squared error in standardized responses, subject to bounds on
the fitted function and its input derivatives:
\[
\begin{aligned}
\min_{f\in\mathcal F_{\mathrm{NN}}}\quad
&\frac{1}{|\mathcal I|}
  \sum_{i\in\mathcal I}\bigl(\widetilde y_i-f(u_i)\bigr)^2,\\
\text{subject to}\quad
&|f(u)|\le32,
&&u\in\mathcal C,\\
&|\partial_u^\nu f(u)|\le32,
&&u\in\mathcal C,\quad 1\le|\nu|\le s.
\end{aligned}
\]
Here $\mathcal I$ indexes the sites used for fitting,
$\widetilde y_i$ is the standardized simulated response at site
$i$, and $\mathcal C$ is the finite set of numerical constraint
points. During validation, $\mathcal I$ excludes the held-out
sites; for the final fit, it includes all training sites.

The objective measures how well the network fits the simulated
responses. The constraints limit the magnitude of the fitted
function and its derivatives. Here $\nu_j$ specifies the number
of derivatives taken in coordinate $j$, and
$|\nu|=\sum_j\nu_j$, so the constraints include all pure and
mixed partial derivatives through order $s$.
These derivatives are computed from the network itself;
no derivative observations enter either the objective or
the constraints.

In the option experiment, derivatives through order three
are constrained, although only Delta and diagonal Gamma
are reported. In the wireless experiment, the constraints
extend through order three for $d=2,4$ and through order four
for $d=6$, although only gradients are reported.
The higher-order constraints control the smoothness of the
fitted surface.

The bound of 32 is fixed across both experiments and applies
to the standardized function and its derivatives with respect
to standardized inputs. Predictions are returned to the original
response units by multiplying by the response standard deviation
and adding the response mean. Derivatives are multiplied by
the response standard deviation and divided by the appropriate
input scaling factors according to the chain rule.

The constraint set $\mathcal C$ is enlarged during training
when violations are detected. Numerical checks use 8192 Sobol
points \citep{sobol1967} in $[-1,1]^d$, all corners, and 128 points on each
boundary face; searches for further violations supply
additional constraint points.

\paragraph{Initialization.}
For initialization, we generate $\tanh$ features with random
weights and fit a linear combination of these features to the
responses by ridge regression.
Feature directions are sampled uniformly on the unit sphere,
and feature centers are sampled from the current fitting inputs. The candidate feature scales are
$\{0.5,1,2,4\}$. The scale multiplies each feature's linear
argument and controls how rapidly it varies with the inputs.
The feature matrix is centered and divided by the square root
of the first hidden-layer width. GCV selects the feature scale
and ridge shift using the standardized responses from the fitting
subset. The 15 candidate ridge shifts are
$10^{-9},10^{-8.5},\ldots,10^{-2}$ times the trace of the
resulting feature Gram matrix.

This regression fit provides the network's initial parameters.
We then refit the output layer subject to the function and
derivative bounds, obtaining the initial feasible network,
which we call stage zero.

\paragraph{Training, stopping, and refitting.}
Subsequent training updates the hidden-layer parameters while
refitting the output layer by constrained least squares.
Each training round allows up to 25 L-BFGS iterations \citep{liunocedal1989} for the
hidden-layer updates and checks the bounds at additional points;
detected violations are added to the constraints.
We save the resulting network after each of at most four rounds.
The initial feasible network is also retained as a candidate,
so response validation may select it without further
hidden-layer training.

We compare stage zero and the subsequent saved stages by their
mean-squared response error on the held-out sites, and select
the stage with the lowest error. We then repeat initialization
and fitting on all training sites, recomputing the response
standardization and using the selected number of rounds as the
training limit. Thus GCV selects the initialization parameters,
while held-out response validation selects how long to train.
The architecture rule, constraint bound, initialization grids,
and training limits are common to both experiments.

This is a numerical approximation to the theoretical constrained
estimator. Local optimization need not attain the global
least-squares minimum, and checks at finitely many points do
not certify the global Sobolev bound assumed in the theory.

\subsection{Parameter settings}\label{app:b-design}

All five methods use the training box $[50,150]^d$, which extends beyond
the test region $[75,125]^d$ and supplies observations on both sides
of its boundary. The global models use scrambled Halton designs
generated by SciPy's \texttt{scipy.stats.qmc.Halton}, with
\texttt{scramble=True}, following the randomization described by
\citet{owen2017halton}. These designs spread sites across the training
box. KR and LPR use Cartesian midpoint grids with $M$ points per
coordinate, $x_k=50+100(k-\tfrac12)/M$. Their denser designs provide
the nearby observations needed for local fitting; the global models
combine information across the entire training sample.
Table~\ref{tab:asian-site-allocation} gives the site counts.
Each budget is divided as evenly as possible among sites, using
$\lfloor B/n\rfloor$ or $\lceil B/n\rceil$ observations per site.
Thus the allocations balance spatial coverage against simulation noise
in each site mean.

All calculations use double precision, with deterministic GPU
settings for NN. The accompanying code provides the random
seeds and solver tolerances.

\begin{table}[tbhp]\centering\small
\caption{Training-site allocations for the option example. Global counts apply to KRR, MKL, and NN.}
\label{tab:asian-site-allocation}
\begin{tabular}{rrrrr}\toprule
$d$ & $B$ & Global & KR & LPR \\\midrule
1 & 50,000 & 50 & 834 & 834 \\
1 & 500,000 & 500 & 1077 & 1159 \\
1 & 5,000,000 & 1000 & 1392 & 1610 \\
2 & 50,000 & 50 & 324 & 289 \\
2 & 500,000 & 500 & 484 & 529 \\
2 & 5,000,000 & 1000 & 784 & 900 \\
3 & 50,000 & 50 & 1331 & 1331 \\
3 & 500,000 & 500 & 2744 & 2744 \\
3 & 5,000,000 & 1000 & 4913 & 6859 \\
4 & 50,000 & 50 & 2401 & 2401 \\
4 & 500,000 & 500 & 10000 & 10000 \\
4 & 5,000,000 & 1000 & 20736 & 28561 \\
\bottomrule\end{tabular}\end{table}

\section{Wireless Communication Network Example}
\label{app:wireless}

\subsection{Simulation Model Building}

We construct a semi-analytical approximation of selected mechanisms in
a wireless network simulator, used only to illustrate the black-box
setting in Section~\ref{subsec:numerical-wireless}.
The model combines distance-based path loss and shadowing, familiar
features of wireless propagation \citep{tse2005fundamentals}, with a
capped quadratic antenna pattern similar to that of
\citet{kifle2013potential}. The formulas below specify how one simulated
response is generated. The estimators receive only these responses;
the model is not calibrated to a particular operational network.

The service region is the square $\mathcal S=[-L,L]\times[-L,L]$. A user location
$U=(U_1,U_2)$ is sampled uniformly from this region, and two antennas
are located at opposite corners, $b_1=(-L,L)$ and $b_2=(L,-L)$.
Antenna $i$ has transmit power $p_i$, azimuth angle $\alpha_i$ (its horizontal orientation), and downtilt angle  $\beta_i$ (the downward direction
of its antenna beam). Thus the six design variables are
$\theta=(p_1,\alpha_1,\beta_1,p_2,\alpha_2,\beta_2)^\top$.
Transmit powers are specified in dBm, a logarithmic scale relative to
one milliwatt. Angles are reported in degrees below for readability,
but all calculations in the simulator use radians.

For a user at $U$, the horizontal range and three-dimensional distance
from antenna $i$ are
\[
 r_i(U)=\|U-b_i\|_2,
 \qquad d_i(U)=\sqrt{r_i(U)^2+h_i^2},
\]
where $h_i$ is the antenna height above the user plane.
The user's horizontal direction and downward elevation angle, viewed
from antenna $i$, are
\[
 \phi_i(U)=\operatorname{atan2}(U_2-b_{i2},U_1-b_{i1}),
 \qquad \psi_i(U)=\operatorname{atan2}(h_i,r_i(U)).
\]
The angular offsets from the antenna's pointing direction are
\[
 \Delta\phi_i=\operatorname{wrap}_{[-\pi,\pi)}
       \{\phi_i(U)-\alpha_i\},
 \qquad \Delta\psi_i=\psi_i(U)-\beta_i.
\]
Here $\operatorname{atan2}$ is the two-argument arctangent, which
retains the direction of the displacement. Wrapping expresses the
horizontal angular difference modulo $2\pi$ in $[-\pi,\pi)$.

The antenna sends a stronger signal toward users close to its pointing
direction. We model the resulting gain by
\[
 G_i(\theta,U)=G_{\max}
 -\min\left\{
  12\left(\frac{\Delta\phi_i}{\phi_{\rm3dB}}\right)^2
  +12\left(\frac{\Delta\psi_i}{\psi_{\rm3dB}}\right)^2,
  A_{\max}\right\}.
\]
Here $G_{\max}$ is the peak gain, $\phi_{\rm3dB}$ and $\psi_{\rm3dB}$
control the horizontal and vertical beam widths, and $A_{\max}$ limits
the attenuation away from the main beam. Signal strength also decreases
with distance. The path loss and received power are
\begin{align*}
 L_i(U)&=L_0+10\eta\log_{10}d_i(U),\\
 P_i^{\rm dBm}(\theta,\xi)&=p_i+G_i(\theta,U)-L_i(U)+Z_i,\\
 P_i(\theta,\xi)&=10^{P_i^{\rm dBm}(\theta,\xi)/10}
                  \quad\text{mW}.
\end{align*}
Distances are in metres; $L_0$ is the loss at one metre and $\eta$ is
the path-loss exponent. The independent zero-mean Gaussian variables
$Z_1,Z_2$, also independent of $U$, represent shadowing on the dB scale.
Gaussian shadowing on the dB scale corresponds to a lognormal
multiplicative effect on power. Thus each simulation draws
$\xi=(U,Z_1,Z_2)$ and computes the two received powers.

The user connects to the antenna with the stronger received signal;
the other antenna contributes interference. With background noise
power $\epsilon_0$, the signal-to-interference-plus-noise ratio and
response are
\begin{align*}
 \operatorname{SINR}(\theta,\xi)
 &=\frac{\max\{P_1(\theta,\xi),P_2(\theta,\xi)\}}
 {\epsilon_0+\min\{P_1(\theta,\xi),P_2(\theta,\xi)\}},
 \\
 F(\theta,\xi)&=\log\operatorname{SINR}(\theta,\xi),
\end{align*}
where $\log$ denotes the natural logarithm. Averaging these responses over user locations
and shadowing gives the performance surface
$J(\theta)=\mathbb E[F(\theta,\xi)]$ whose gradient is estimated in
Section~\ref{subsec:numerical-wireless}.

\subsection{Benchmark method}

For the central-FD benchmark in Section~\ref{subsec:numerical-wireless},
the perturbation half-steps are $1$ dB for power, $5^\circ$ for azimuth,
and $1^\circ$ for downtilt. The $10^5$ observations per query are divided as evenly as
possible among the $2d$ perturbed settings, with independent samples for positive and negative perturbations.
Each setting receives 25,000 observations in 2D, 12,500 in 4D, and either 8,333 or 8,334 in 6D. Reference gradients use one-quarter of these half-steps and $10^6$ observations per perturbed setting. Common random numbers within each reference pair reduce the
noise in its difference. These separately computed references are used
only for evaluation and remain subject to Monte Carlo and
finite-difference approximation error.

\subsection{LTD estimator}

We follow the five fitting procedures described in
Appendix~\ref{app:local-protocol}, with the following adjustments
for the wireless example. All five methods normalize each active
input to $[-1,1]$ over its physical training range, so powers and
angles enter the fits on comparable scales.
KRR and MKL use ordinary response GCV instead of robust GCV,
with all training sites used for both length-scale and
regularization selection. KR and LPR retain GCV selection on
at most 200 training sites and use all sites for the final fit.
Their bandwidth safeguards use $r_*=1$, since only first
derivatives are estimated.

MKL retains the three full Gaussian kernels, normalization
using derivative orders one through three, and the lower bound
of $1/6$ on each mixture weight.
All five methods learn from the full vector of active inputs
and simulated responses, without incorporating additional
knowledge of the simulator's internal structure into the
fitted models.

NN uses the architecture rule, response standardization,
GCV initialization, and response-validation stopping described
in Appendix~\ref{app:b-nn}. The constraint bound remains 32
in standardized units, with derivatives constrained through
order three for $d=2,4$ and through order four for $d=6$.
The finite-point constraints and local optimization have the
same limitations discussed there.
All estimated derivatives are converted back to the original
input and response units; angular derivatives are with respect
to radians.

\begin{table}[htbp]
\centering\small
\caption{Fixed parameters in the wireless simulator.}
\label{tab:wireless-fixed-parameters}
\begin{tabular}{lll}
\hline
Quantity & Symbol & Value \\
\hline
Region half-width & $L$ & $500\,\mathrm{m}$ \\
Antenna height & $h_1=h_2$ & $30\,\mathrm{m}$ \\
Peak antenna gain & $G_{\max}$ & $17\,\mathrm{dBi}$ \\
Horizontal 3-dB beamwidth & $\phi_{\rm 3dB}$ & $65^\circ$ \\
Vertical 3-dB beamwidth & $\psi_{\rm 3dB}$ & $10^\circ$ \\
Maximum pattern attenuation & $A_{\max}$ & $30\,\mathrm{dB}$ \\
Path loss at $1\,\mathrm{m}$ & $L_0$ & $38.5\,\mathrm{dB}$ \\
Path-loss exponent & $\eta$ & $3.2$ \\
Shadowing distribution & $Z_i$ & $N(0,6^2)\,\mathrm{dB}$ \\
Background noise & $10\log_{10}\epsilon_0$ & $-104\,\mathrm{dBm}$ \\
\hline
\end{tabular}
\end{table}

\subsection{Parameter settings}
\label{app:c-design}

Table~\ref{tab:wireless-fixed-parameters} lists the fixed simulation
parameters. The square has side length $1\,\mathrm{km}$, which is large
enough to produce substantial distance variation while retaining inexpensive
simulation. 
A $30\,\mathrm{m}$ antenna height and $43\,\mathrm{dBm}$ nominal
power are representative of a macro-cell-style toy example.  The $17\,\mathrm{dBi}$
peak gain, $65^\circ$ horizontal beamwidth, $10^\circ$ vertical beamwidth,
and $30\,\mathrm{dB}$ attenuation cap create smooth directional variation
near the main lobe without allowing unrealistic attenuation in the back
lobe.  The path-loss exponent $3.2$ represents an obstructed environment,
and $6\,\mathrm{dB}$ shadowing gives meaningful output noise.  Finally,
$-104\,\mathrm{dBm}$ background noise makes the example predominantly
interference-limited without removing the effect of transmit power.
Overall, these parameter choices define an
illustrative test problem.

The training ranges are
$p_1,p_2\in[38,48]$ dBm,
$\alpha_1\in[-75,-15]^\circ$,
$\alpha_2\in[105,165]^\circ$, and
$\beta_1,\beta_2\in[2,18]^\circ$.
At the centre of this box, the two antennas point along opposite
directions of the diagonal joining their locations.
In the lower-dimensional experiments, only the active coordinates vary;
the remaining coordinates are taken from
\(
 \theta_{\rm ref}=
 (40.5,-23.5^\circ,12^\circ,45.5,115.5^\circ,6.5^\circ)^\top.
\)
The two-dimensional problem varies the two powers, the four-dimensional
problem varies the four angles, and the six-dimensional problem varies
all coordinates.

KRR, MKL, and NN share centered Latin hypercube training sites
and simulated responses within each replication. KR and
LPR use scrambled Halton sites, which allow their site counts to increase
with budget without constructing a full tensor grid. All five methods
use the same physical training ranges.
Table~\ref{tab:wireless-site-allocation} gives the site counts, and each
simulation budget is divided as evenly as possible across sites.
The test ranges are inset by the FD half-steps:
$[39,47]$ dBm for power, $[-70,-20]^\circ$ and $[110,160]^\circ$ for
azimuth, and $[3,17]^\circ$ for downtilt. This keeps every FD perturbation
inside the training box. For KRR and MKL, the kernel length-scale and regularization
grids are those in Appendix~\ref{app:b-kernel}; tuning uses all training sites and the current replication's simulated responses, without additional pilot simulations or derivative
observations.

\begin{table}[htbp]
\centering\small
\caption{Training-site allocations for the wireless example. Global counts apply to KRR, MKL, and NN.}
\label{tab:wireless-site-allocation}
\begin{tabular}{rrrrr}\toprule
$d$ & $B$ & Global & KR & LPR \\\midrule
2 & 100,000 & 200 & 121 & 121 \\
2 & 1,000,000 & 200 & 192 & 216 \\
2 & 10,000,000 & 200 & 304 & 383 \\
4 & 100,000 & 1000 & 625 & 625 \\
4 & 1,000,000 & 1000 & 1347 & 1570 \\
4 & 10,000,000 & 1000 & 2901 & 3944 \\
6 & 100,000 & 1200 & 729 & 729 \\
6 & 1,000,000 & 1200 & 1956 & 2306 \\
6 & 10,000,000 & 1200 & 5247 & 7290 \\
\bottomrule\end{tabular}\end{table}

\end{document}